\PassOptionsToPackage{table}{xcolor}
\documentclass{article}
\usepackage{iclr2027_conference,times}
\usepackage[T1]{fontenc}

\usepackage{amsmath,amsfonts,bm}

\def\eqref#1{equation~\ref{#1}}

\def\1{\bm{1}}

\DeclareMathAlphabet{\mathsfit}{\encodingdefault}{\sfdefault}{m}{sl}
\SetMathAlphabet{\mathsfit}{bold}{\encodingdefault}{\sfdefault}{bx}{n}

\newcommand{\R}{\mathbb{R}}

\newcommand{\OracleGreedySurplus}{0.012}
\newcommand{\OracleUnopSurplus}{0.145}
\newcommand{\OracleWelfareSurplus}{0.162}
\newcommand{\OracleUnopDiff}{0.133}
\newcommand{\OracleUnopLo}{0.124}
\newcommand{\OracleUnopHi}{0.141}
\newcommand{\OracleUnopPos}{89.2}
\newcommand{\OracleUnopBayes}{0.946}

\newcommand{\OracleN}{120}
\newcommand{\LockBayes}{0.250}
\newcommand{\LockAcc}{0.248}
\newcommand{\LockKL}{0.012}
\newcommand{\LockHuman}{0.860}
\newcommand{\LockGreedyHuman}{0.750}
\newcommand{\LockGap}{0.0000}
\newcommand{\RecExactGreedy}{0.082}
\newcommand{\RecExactUnop}{0.160}
\newcommand{\RecVanillaGreedy}{0.495}
\newcommand{\RecVanillaUnop}{0.459}

\newcommand{\NoiseZero}{0.145}
\newcommand{\NoiseThirty}{-0.006}

\newcommand{\FamOracleG}{0.011}
\newcommand{\FamOracleU}{0.151}
\newcommand{\FamThompG}{0.010}
\newcommand{\FamThompU}{0.107}
\newcommand{\FamUcbG}{0.004}
\newcommand{\FamUcbU}{0.006}
\newcommand{\FamExpG}{0.156}
\newcommand{\FamExpU}{0.191}
\newcommand{\FamOracleKL}{0.0043}
\newcommand{\RecExactDiff}{0.078}
\newcommand{\RecExactLo}{0.064}
\newcommand{\RecExactHi}{0.092}
\newcommand{\RecExactPos}{74.4}
\newcommand{\RecBeamDiff}{0.102}
\newcommand{\RecBeamLo}{0.090}
\newcommand{\RecBeamHi}{0.115}
\newcommand{\RecBeamPos}{98.8}
\newcommand{\RecVanDiff}{-0.036}
\newcommand{\RecVanLo}{-0.040}
\newcommand{\RecVanHi}{-0.032}
\newcommand{\BoltMatch}{0.1454}
\newcommand{\UnopMatch}{0.1453}

\newcommand{\DenseUnop}{0.146}
\newcommand{\DenseGreedy}{0.013}
\newcommand{\DenseOracle}{0.163}
\newcommand{\NoiseViol}{0.119}
\newcommand{\ConflictNeg}{0.104}
\newcommand{\ConflictNegLo}{0.073}
\newcommand{\ConflictNegHi}{0.139}
\newcommand{\ConflictPos}{-0.002}
\newcommand{\ConflictPosLo}{-0.009}
\newcommand{\ConflictPosHi}{0.004}

\usepackage[hyphens]{url}
\usepackage{graphicx}
\usepackage{booktabs}
\usepackage{amsmath,amssymb,amsthm}
\usepackage{xcolor}
\usepackage{colortbl}
\usepackage{float}
\usepackage{algorithm}
\usepackage{algpseudocode}
\usepackage{fontawesome5}
\usepackage{hyperref}
\hypersetup{colorlinks=true,citecolor=blue,linkcolor=blue,urlcolor=blue}

\definecolor{cellgreen}{RGB}{186,228,196}
\definecolor{cellblue}{RGB}{189,215,238}

\newtheorem{theorem}{Theorem}
\newtheorem{proposition}{Proposition}
\newtheorem{definition}{Definition}

\newcommand{\unop}{UNOP}
\newcommand{\aunop}{Aware-UNOP}
\newcommand{\cA}{\mathcal{A}}
\newcommand{\cH}{\mathcal{H}}
\newcommand{\cO}{\mathcal{O}}
\newcommand{\EE}{\mathbb{E}}
\newcommand{\PP}{\mathbb{P}}
\newcommand{\kl}{\mathrm{KL}}
\newcommand{\tv}{\mathrm{TV}}

\title{Learnable Randomization as Commitment\\
Against Adaptive Optimizers}

\author{
\textbf{Zihan Deng}$^{1}$\thanks{Corresponding author: \texttt{zihandeng@connect.hku.hk}.},
\textbf{Chuanzhi Xu}$^{2}$,
\textbf{Xiaozhen Zhong}$^{3}$,
\textbf{Haoyang Li}$^{2}$,
\textbf{Junjie Huang}$^{4}$\\[0.45em]
$^{1}$The University of Hong Kong\\
$^{2}$The University of Sydney\\
$^{3}$University of Electronic Science and Technology of China\\
$^{4}$University of Science and Technology of China\\[0.45em]
\href{https://github.com/FrankDengAI/unop}{\faGithub\hspace{0.35em}\textit{Code}}
}

\iclrfinalcopy
\makeatletter
\def\@maketitle{%
 \vbox{\hsize\textwidth
  \centering
  {\LARGE\bfseries \@title\par}
  \vskip 0.75em
  {\large \@author\par}
  \vskip 0.45em
  {\footnotesize \@thanks\par}
  \vskip 0.2in
 }%
}
\makeatother

\begin{document}

\maketitle
\fancyhead{}
\renewcommand{\headrulewidth}{0pt}
\pagestyle{plain}

\begin{abstract}
A pricing page can walk the posted price up to the last amount a buyer still accepts, a recommender can hold back a better item for a barely acceptable promoted one, and a classifier can shift its boundary once applicants change their features. The system predicts the response and then picks the menu that serves its own objective, so the surplus above the user's cutoff is taken. Playing the single best action publishes that cutoff, while noise on actions the user would never take throws away payoff and teaches the platform that a worse menu is still acceptable.
We study unpredictable near-optimal policies (\unop{}), which mix uniformly on near-best actions that remain individually rational. The mixture is a commitment about the response. On a finite price grid, when the best sure-demand price strictly out-earns the randomized band, a seller who already knows the curve posts below the band, and the purchase that occurs is deterministic. Knowing that curve is not the same as predicting the next draw. The mixture can be learned and the optimizer can match its best response, while the user's payoff stays higher because the mixture changes which action is targeted.
In pricing and in policy-aware recommendation this leaves more surplus than greedy play when the platform optimizes against the curve and more than one action is acceptable. The gain goes away under quality ranking, a singleton near-optimal set, a wrong utility estimate, or a short-horizon explorer. That is also where mixing should be turned off if the other side is trying to cooperate.

\end{abstract}

\section{Introduction}
\label{sec:intro}

\begin{figure}[t]
\centering
\vspace{-0.3em}
\includegraphics[width=\linewidth]{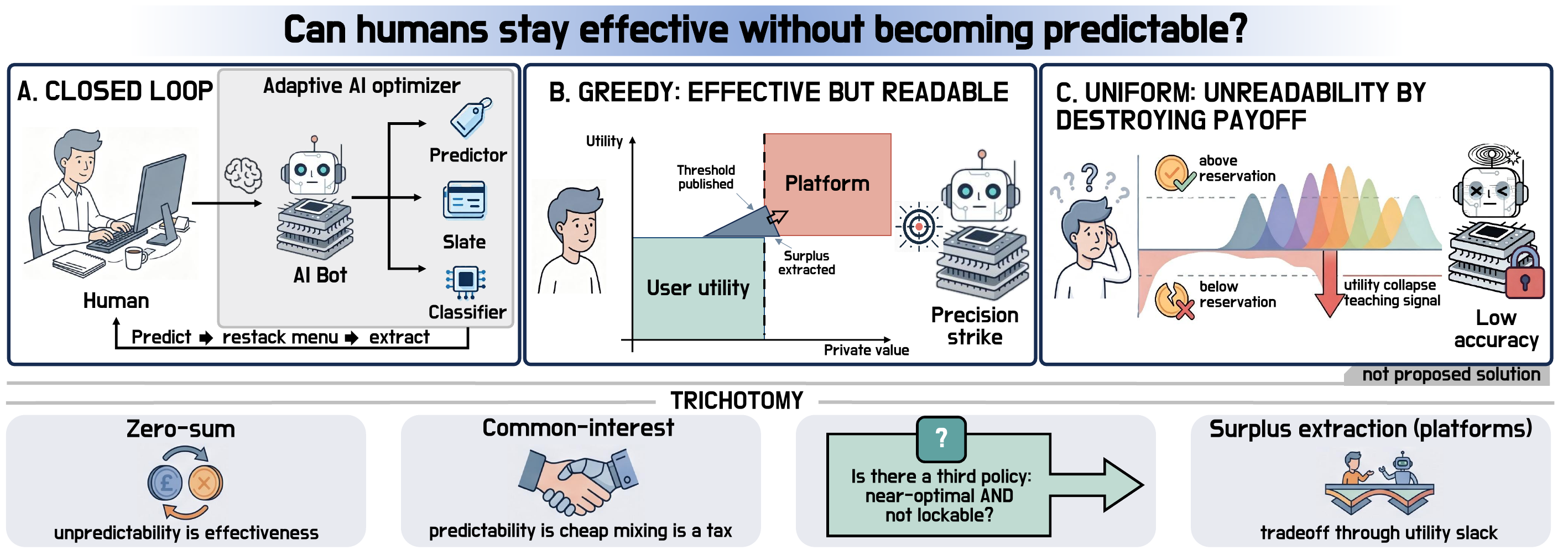}
\caption{(A)~A human and an adaptive optimizer in a closed loop. The optimizer may learn the response distribution and still revise the menu. (B)~A greedy threshold publishes the user's cutoff. (C)~Uniform noise refuses good actions and destroys payoff. The protective randomness in pricing is the response the buyer would give on menus the optimizer then declines to offer. A mixed security strategy is required in zero-sum play, while mixing is a coordination tax when interests coincide.}
\label{fig:problem}
\end{figure}

A growing fraction of human decisions are taken next to an adaptive algorithm, whether a platform sets a price after watching who buys~\citep{calvano2020artificial,braverman2018selling}, a recommender builds a list after watching who clicks~\citep{dean2020preference,hron2023modeling,stray2024building}, or a classifier is retrained on the features applicants choose to present~\citep{hardt2016strategic,perdomo2020performative}. In each case the algorithm is an optimizer rather than a fixed environment, because it predicts behavior and then chooses the action that maximizes its own objective given that prediction. Once the prediction is accurate the menu moves with it, so the user can keep accepting the offered price, list, or decision while the surplus above the user's cutoff is taken by the optimizer.

Once the interaction is written as an ordinary learning problem, it is easy to study only how to predict the user or how to extract surplus from a user who has already been predicted. Strategic buyers are usually treated as long-horizon manipulators of a learning seller, or the seller is designed so that the manipulation does not pay~\citep{drutsa2018weakly,drutsa2020optimal,huh2026strategy}. A pure best response is a deterministic function of the current menu, so an optimizer that identifies it can post a price just below value, withhold a better organic item in favor of a barely acceptable promoted one, or shift a decision boundary against a known gaming direction. Uniform randomization conceals that function only by collapsing the user's own payoff, while $\varepsilon$-greedy and Boltzmann exploration put noise on actions that are far from optimal or individually irrational, and an adaptive optimizer can read those loss-making plays as evidence that the current menu is still acceptable.

We do not try to hide the response distribution, and the pricing guarantee below assumes the optimizer knows it exactly. Figure~\ref{fig:problem} shows the closed loop. In pricing, once the seller stays below the randomized band, the action that is actually played is often deterministic, because the randomization only specifies what the buyer would do at prices the seller then chooses not to post. We call this the response-curve commitment effect. The platform can learn the distribution and optimize against the whole curve, yet still be led to a menu that takes less surplus. The question is therefore not how to keep the policy secret. It is which response curve the user can credibly commit to once the platform is free to learn it.

The shape of the user's payoffs decides how much room that commitment has. Let $u$ be the estimated payoff vector and let $\cA_\alpha(u)=\{a:u(a)\ge \max u-\alpha\}$ be the set of actions that never lose more than $\alpha$ relative to the best one-step action. When the set is a singleton, effectiveness forces a predictable action, whereas a larger set can be mixed so as to maximize entropy or to put less weight on actions the current forecast already expects. We call the resulting family unpredictable near-optimal policies (\unop{}). Individual rationality is part of the support, because mixing into a strictly loss-making action, such as buying above value or clicking below reservation, teaches the optimizer that the cutoff can still be moved. \unop{} therefore mixes only inside $\cA_\alpha(u)$ and above a reservation floor.

The same payoff shape produces three different conclusions, which we prove and then measure. In zero-sum games such as matching pennies and rock, paper, scissors, an adaptive best responder turns a greedy pure policy into a security-level loss, so unpredictability is what makes the user effective~\citep{nash1950equilibrium,osborne1994course}. In common-interest games such as coordination and stag hunt, the Pareto-optimal pure profile is already effective and fully predictable, so mixing is only a coordination tax and being read is what good play looks like when the other side is a collaborator~\citep{carroll2019utility,hadfield2016cooperative}. In surplus extraction the platform identifies a threshold and sits on the user's cutoff~\citep{hardt2022performative}. A greedy user gives that cutoff away and a uniform user refuses both the cutoff and the surplus, while \unop{} blurs only the boundary, with sure demand below $v-\alpha$, sure refusal above $v$, and mixing in between. Under Proposition~\ref{prop:price} this preserves $\Omega(\alpha)$-scale pricing surplus, and the same support restriction forces a margin-seeking list to trade acceptance against user quality.

Our contributions can be summarized as follows:
\begin{itemize}
\item We separate distributional learning error, realized-action predictability, optimizer regret relative to an exact best response, and the surplus that best response extracts.
\item We characterize constant-band commitments on a finite price grid. \unop{} is the maximum-entropy support-safe response. It induces the same posted price as the welfare-maximizing band rate precisely when purchase rate $1/2$ strictly keeps the seller below the band.
\item We show that predictors can learn the mixture and match an oracle best response, while user surplus stays above greedy play in pricing and under a policy-aware recommender. Quality ranking reverses that comparison, a singleton $\alpha$-set on every feasible menu removes it, and large utility noise or a short-horizon explorer makes the protection fail.
\end{itemize}

\section{Related work}
\label{sec:related}

Learning in Games and Strategizing against Learners.
Classical learning in games studies no-regret dynamics and their equilibria~\citep{fudenberg1998theory,cesabianchi2006prediction}. Weighted-majority and exponential-weights bounds make the same point~\citep{littlestone1994weighted,freund1997decision}. A recent line asks the converse: how should a patient optimizer play against a no-regret learner~\citep{deng2019strategizing,braverman2018selling,mansour2022strategizing,camara2020mechanisms,haghtalab2022learning,kolumbus2022how}? The typical conclusion is that the optimizer can secure at least the Stackelberg value, and sometimes more against mean-based learners. Repeated posted-price auctions with a strategic buyer, and seller algorithms designed to remain robust to that buyer, study the other side of the interaction as long-horizon manipulation or as a robust mechanism~\citep{drutsa2018weakly,drutsa2020optimal,huh2026strategy}. We study a simple support-constrained commitment and separate policy learning, realization prediction, optimizer regret, and extraction. Proper scoring rules make the same separation at the level of a forecast: log loss equals entropy plus a divergence, so a high log loss need not mean the forecast missed the distribution~\citep{cover2006elements}.

Performative Prediction and Platforms.
Performative prediction~\citep{perdomo2020performative,hardt2022performative} and strategic classification~\citep{hardt2016strategic,bruckner2012stackelberg,dong2018strategic,milli2019social,zrnic2021who,levanon2021strategic,chen2020learning,shavit2020causal,kleinberg2019how} study distribution shift caused by agent responses. The usual goal is a stable or Stackelberg-optimal predictor. Agent welfare, when it appears, is often a constraint on the institution~\citep{milli2019social,levanon2021strategic}. We focus on the human side. The same literature, with recommender choice-set design~\citep{dean2020preference,benporat2018game,hron2023modeling,ie2019reinforcement,schnabel2016recommendations} and algorithmic pricing~\citep{calvano2020artificial,auer2002finite}, is where surplus extraction lives. Performative power~\citep{hardt2022performative} quantifies how much a platform can steer. We treat that steering as sitting on a predicted threshold.

Entropy-Regularized Play and Coordination.
Maximum-entropy {IRL} and soft {RL}~\citep{ziebart2008maximum,haarnoja2018soft,jaynes1957information}, quantal response~\citep{mckelvey1995quantal}, and information theory~\citep{cover2006elements} all justify Gibbs distributions. We use them as the Shannon-Pareto baseline. For a human facing extraction the natural constraint is worst-case utility loss, and the maximum-entropy solution is uniform on the actions within $\alpha$ of the best payoff. Cooperative {IRL}~\citep{hadfield2016cooperative,carroll2019utility,russell2019human} instead emphasizes making humans easy to predict. We agree when the other side is a collaborator, and we locate where that advice reverses.

\section{Method}
\label{sec:method}

\subsection{Interaction and objectives}
We consider a $T$-round interaction between a human $\cH$ and an adaptive optimizer $\cO$. At round $t$, a context $x_t$ is realized and the history $h_{t-1}=(x_s,a_s,o_s)_{s<t}$ is public. The human chooses $a_t\in A$ from a finite action set and the optimizer chooses $o_t\in O$. Instantaneous payoffs are $u_H(a_t,o_t,x_t)$ and $u_O(a_t,o_t,x_t)$. The optimizer maintains a predictive distribution $P_t(\cdot\mid x_t,h_{t-1})$ on $A$, plays a best response that may still explore
\begin{equation}
\label{eq:opt-br}
o_t \in \arg\max_{o\in O}\ \EE_{a\sim P_t}\bigl[u_O(a,o,x_t)\bigr]
\end{equation}
or a bandit version of~\eqref{eq:opt-br}, such as UCB on posted prices or unique-best recommendation lists, and updates $P_t$ from $a_t$. In simultaneous-move games the human does not see the optimizer's current mix, and forms
\begin{equation}
\label{eq:mean-based}
\hat u_t(a) \;=\; \sum_{o\in O} \hat\sigma_t(o)\, u_H(a,o,x_t),
\end{equation}
where $\hat\sigma_t$ is the empirical mix of past optimizer actions in simultaneous bimatrix play. In sequential menus, namely pricing and recommendation, the human observes the posted $o_t$ before acting, so $\hat u_t$ is the displayed utility of the current menu rather than an average over past menus. Then $a_t\sim \pi_t(\cdot\mid \hat u_t,P_t)$.

\begin{definition}
The human's effectiveness is $V_H=\frac1T\sum_{t=1}^T \EE[u_H(a_t,o_t,x_t)]$. In surplus-extraction domains we report the analogous surplus over a reservation utility.
\end{definition}

\begin{definition}
Let $\pi$ be the user's mixed policy and $P$ a predictive distribution, with probabilities clipped at $10^{-12}$ before logarithms. Distributional learning error is $L(\pi,P)=\kl(\pi\|P)$, and realization predictability is the Bayes accuracy $R(\pi)=\max_a\pi(a)$, which upper-bounds every predictor's one-step hit rate. Write $V_O^\star(\pi)=\max_o\EE_{a\sim\pi}[u_O(a,o)]$ for the optimizer's exact best-response value and $o(P)$ for the action it plays from $P$. Optimizer regret is the oracle gap $G(\pi,P)=V_O^\star(\pi)-V_O(\pi,o(P))$, while extraction is $V_O^\star(\pi)$ itself, or the user's surplus loss relative to a menu that is not chosen to extract. A zero gap means only that the learned optimizer has reached the oracle best response, not that the policy cannot be exploited.
\end{definition}

\begin{proposition}
\label{prop:decomp}
For any $\pi$ and $P$ in the simplex, the best one-step hit rate satisfies $\max_P\PP_{a\sim\pi}(\arg\max P=a)=\max_a\pi(a)$. The expected log loss decomposes as $\EE_{a\sim\pi}[-\log P(a)]=H(\pi)+\kl(\pi\|P)$, so the excess over the entropy is zero if and only if $P=\pi$ on the support of $\pi$. Raw log loss and raw top-1 accuracy therefore cannot, by themselves, show that a predictor failed to learn the mixture.
\end{proposition}

These are operational measures: they are what an adaptive {AI} system computes. We do not equate unpredictability with differential privacy~\citep{dwork2006differential,duchi2013local}. Write $u\in\R^{|A|}$ for the current expected-utility vector, $u^\star=\max_a u(a)$, and
\begin{equation}
\label{eq:alpha-set}
\cA_\alpha(u)\;=\;\{a\in A:u(a)\ge u^\star-\alpha\}.
\end{equation}
Every action in $\cA_\alpha(u)$ loses at most $\alpha$ relative to the best one-step action. If $|\cA_\alpha(u)|=1$, any policy supported on that singleton is a $\delta$-peak with $\delta=1$, aside from how ties are broken, and effectiveness forces predictability. Appendix~\ref{app:extra} measures robustness when $\hat u$ is a noisy version of $u$.

\subsection{Three cases and two tradeoffs}
\label{sec:theory}

\begin{proposition}
\label{prop:zerosum}
In a finite zero-sum game, any pure action whose security level is below the game value $v^\star$ is exploitable by an exact best responder. In matching pennies, every deterministic policy has value $-1$. More generally, let $\mathcal{S}=\{\pi\in\Delta(A):\min_o u_H(\pi,o)\ge v^\star\}$ be the security-strategy polytope. There is a game-dependent constant $\kappa_G<\infty$ such that every policy with value at least $v^\star-\varepsilon$ satisfies
\[
\inf_{\pi^\star\in\mathcal{S}}\tv(\pi,\pi^\star)\le \kappa_G\varepsilon.
\]
Thus, in games such as matching pennies and rock, paper, scissors, whose security set contains no peaked policy, sufficiently effective play must remain non-peaked. The constant and the conclusion are game dependent. Zero-sum structure alone does not imply that every game has a unique or fully mixed security strategy.
\end{proposition}

The proof is standard security-level reasoning~\citep{osborne1994course}. We record it in Appendix~\ref{app:proofs}. Empirically, greedy obtains mean payoff $-0.143$ in our zero-sum suite, while \unop{} sits at $0.000$ and \aunop{} at $+0.054$.

\begin{proposition}
\label{prop:common}
Suppose $u_H=u_O=u$ and $(a^\star,o^\star)$ is the unique global payoff-maximizing pure profile. Let $u^\star=u(a^\star,o^\star)$ and $\bar u=\max_{a\ne a^\star,o}u(a,o)<u^\star$. Then the effectiveness-maximizing human policy is a point mass on $a^\star$, paired with $o^\star$. For any mixture with $\pi(a^\star)=p<1$, and for any partner action, including a best response, expected common payoff is at most $p u^\star+(1-p)\bar u$, hence it pays at least the linear coordination tax $(1-p)(u^\star-\bar u)$. A Shannon constraint $H(\pi)\ge h>0$ therefore forces $p\le p^\star(h)<1$.
\end{proposition}

In our common-interest suite, greedy, \unop{} with small $\alpha$, and \aunop{} all achieve payoff $2.598$, with predictive accuracy $1.0$. Uniform falls to $1.762$.

The two natural ways of asking a policy to be random but not too costly are not interchangeable. One constrains every atom of the support. The other constrains only the mean.

\begin{theorem}
\label{thm:frontiers}
Fix $u\in\R^{|A|}$ and write $\Delta(A)$ for the simplex. Consider
\begin{align}
\label{eq:wc}
(\mathrm{WC}_\alpha)\qquad
&\max_{\pi\in\Delta(A)}\ H(\pi)
\quad\text{s.t.}\quad
\pi(a)=0 \text{ whenever } u(a)<u^\star-\alpha,\\
\label{eq:sh}
(\mathrm{SH}_h)\qquad
&\max_{\pi\in\Delta(A)}\ \pi\cdot u
\quad\text{s.t.}\quad
H(\pi)\ge h.
\end{align}
\begin{enumerate}
\item Per-action worst-case loss. $(\mathrm{WC}_\alpha)$ is uniquely solved by $\pi=\mathrm{Unif}(\cA_\alpha(u))$, with value $H=\log|\cA_\alpha(u)|$ and $\EE_\pi[u]\ge u^\star-\alpha$. Conversely, a policy places positive mass outside $\cA_\alpha$ if and only if it violates the support-wise guarantee $\max_{a\in\mathrm{supp}(\pi)}(u^\star-u(a))\le\alpha$. This converse concerns the loss of every action that may be played. A policy with very small mass outside $\cA_\alpha$ may still have expected loss below $\alpha$.
\item Shannon constraint. For any $h\in\bigl(H(\mathrm{Unif}(\arg\max u)),\log|A|\bigr]$, $(\mathrm{SH}_h)$ is uniquely solved by a Gibbs distribution $\pi_\tau(a)\propto \exp(u(a)/\tau)$ for some $\tau>0$~\citep{jaynes1957information,cover2006elements}. Equivalently, $\pi_\tau$ maximises $\pi\cdot u+\tau H(\pi)$ over $\Delta(A)$. For $h\le H(\mathrm{Unif}(\arg\max u))$, every expected-utility maximizer with entropy at least $h$ is supported on $\arg\max u$.
\end{enumerate}
\end{theorem}

$(\mathrm{WC}_\alpha)$ maximises a strictly concave objective on a simplex face and is therefore uniquely solved. $(\mathrm{SH}_h)$ maximises a linear objective over the strictly convex set of policies whose entropy is at least $h$. Uniqueness of the Gibbs solution on $\bigl(H(\mathrm{Unif}(\arg\max u)),\log|A|\bigr]$ follows from that shape. KKT stationarity for $(\mathrm{WC}_\alpha)$ on the relative interior of $\Delta(\cA_\alpha)$ is $-\log\pi(a)-1+\lambda=0$, so $\pi$ is constant on $\cA_\alpha$. Equivalently, $\mathrm{Unif}(\cA_\alpha)$ is the $I$-projection of $\mathrm{Unif}(A)$ onto $\{\pi:\mathrm{supp}(\pi)\subseteq\cA_\alpha\}$, since $\kl(\pi\|\mathrm{Unif}(A))=\log|A|-H(\pi)$. For $(\mathrm{SH}_h)$, stationarity of $\pi\cdot u+\tau H(\pi)$ is $u(a)-\tau(\log\pi(a)+1)+\mu=0$, hence Gibbs.

The programs coincide only in special cases, when a single $\alpha$-level set is already $A$, or when $h=0$. For any $\tau>0$, $\pi_\tau(a)>0$ for every $a$. If $\cA_\alpha\neq A$, write $\gamma=\min_{a\notin\cA_\alpha}(u^\star-\alpha-u(a))>0$. Matching bounds on the normalizer $Z=\sum_b e^{u(b)/\tau}$ give
\begin{equation}
\label{eq:gibbs-leak}
\frac1{|A|}\exp\bigl(-(\alpha+\gamma)/\tau\bigr)
\;\le\;
\pi_\tau\bigl(A\setminus\cA_\alpha\bigr)
\;\le\;
\frac{|A\setminus\cA_\alpha|}{|\cA_\alpha|}\exp(-\gamma/\tau).
\end{equation}
The leak is acceptable as exploration in {RL}. Against an extractor it is a teaching signal. Unrestricted Boltzmann and $\varepsilon$-greedy put mass outside $\cA_\alpha$, including individually irrational actions. The pricing comparisons in Section~\ref{sec:results} use two different baselines: {IR}-Boltzmann refuses purchases above value but can still leave $\cA_\alpha$, and support-matched Boltzmann is renormalized on $\cA_\alpha^{\mathrm{IR}}$. Appendix~\ref{app:ir} compares \unop{} with and without the reservation floor.

\begin{theorem}
\label{thm:lockin}
Let $P_t$ be the Dirichlet-multinomial posterior mean with a uniform prior on a finite action set, and let actions be drawn i.i.d.\ from a fixed $\pi$. Then $P_t\to\pi$ almost surely, so $\kl(\pi\|P_t)\to 0$ on the support of $\pi$ and $\max_a P_t(a)\to\max_a\pi(a)$. If $\pi$ is a $\delta$-peak with $\delta>1/2$, the posterior mode equals that action after $O\bigl(\log(1/\eta)/(\delta-1/2)^2\bigr)$ observations with probability at least $1-\eta$. If $\pi$ is uniform on $k$ actions, the same convergence leaves $\max_a P_t(a)\to 1/k$. The mixture is identified, and the one-step hit rate remains bounded by the policy. The proof records how the visit counts grow.
\end{theorem}

\begin{proposition}
\label{prop:price}
Let $\mathcal P$ be a price grid with spacing $\Delta$ and let the buyer use the constant-band response
\begin{equation}
\label{eq:unop-demand}
q_{\alpha,\rho}(p)=\begin{cases}
1 & p< v-\alpha,\\
\rho & v-\alpha \le p \le v,\\
0 & p>v,
\end{cases}
\end{equation}
with $\rho\in(0,1)$. Write $p^-_\alpha(v)=\max\{p\in\mathcal P:p<v-\alpha\}$ and $p^+_\alpha(v)=\max(\mathcal P\cap[v-\alpha,v])$, when those sets are nonempty. A seller who knows $q_{\alpha,\rho}$ and breaks revenue ties by lowest buyer surplus posts in the sure-demand region whenever
\begin{equation}
\label{eq:deter}
p^-_\alpha(v)>\rho\,p^+_\alpha(v).
\end{equation}
The posted price is then $p^-_\alpha(v)$ and buyer surplus is $v-p^-_\alpha(v)\in(\alpha,\alpha+\Delta]$. The guarantee uses the response curve itself and does not use prediction error. \unop{} is the case $\rho=1/2$. Let $r^\star=p^-_\alpha(v)/p^+_\alpha(v)$. Every $\rho<r^\star$ induces that same price, so those rates are exactly the welfare-maximizing constant-band commitments. Binary entropy on $(0,1)$ is maximized at $1/2$, which means that $\rho=1/2$ is the unique entropy-maximizing rate that keeps the seller below the band when $r^\star>1/2$. When $r^\star\le 1/2$ that set of rates, $(0,r^\star)$, is open, so no entropy-maximizing rate exists and entropy only increases as $\rho$ approaches $r^\star$ from below. The inequality $\alpha\le v/2$ is not enough for this conclusion, because if $v=2\alpha$ and $v\in\mathcal P$, then $p^+_\alpha(v)=v$ and $p^-_\alpha(v)<\alpha=v/2$, so~\eqref{eq:deter} fails at $\rho=1/2$ and the seller can post inside the band.
\end{proposition}

Proposition~\ref{prop:price} is a commitment about prices that are not posted. Once the seller stays below the band, the action at the posted price $p^-_\alpha(v)$ is a sure purchase, and the mixture that moved the seller sits on prices that are not posted. On a uniform grid the strict inequality holds for $\rho=1/2$ when $\alpha<v/2-\Delta$, and it can fail on the boundary $\alpha=v/2$.

\subsection{Mixing only where it is safe}
Figure~\ref{fig:method} summarizes the pipeline. Individual rationality is part of the support: mixing into strictly loss-making actions, such as buying above value or clicking below reservation, teaches the optimizer to move the cutoff. Given $u$, slack $\alpha\ge 0$, and reservation $\underline{u}\in\{-\infty\}\cup\R$, define
\begin{equation}
\label{eq:ir-set}
\cA_\alpha^{\mathrm{IR}}(u)
\;=\;
\begin{cases}
\cA_\alpha(u)\cap\{a:u(a)\ge \underline{u}\}
& \text{if that intersection is nonempty},\\
\cA_\alpha(u)
& \text{otherwise}.
\end{cases}
\end{equation}
Exact arithmetic makes the first branch apply whenever the reservation set is nonempty, because $u^\star\ge\underline{u}$ then forces the global maximizer into the intersection. The second branch covers rare numerical cases. The reservation is $0$ for skip in posted pricing, $\tau$ in recommendation, and $-\infty$ in bimatrix games.

\begin{figure}[t]
\centering
\includegraphics[width=\linewidth]{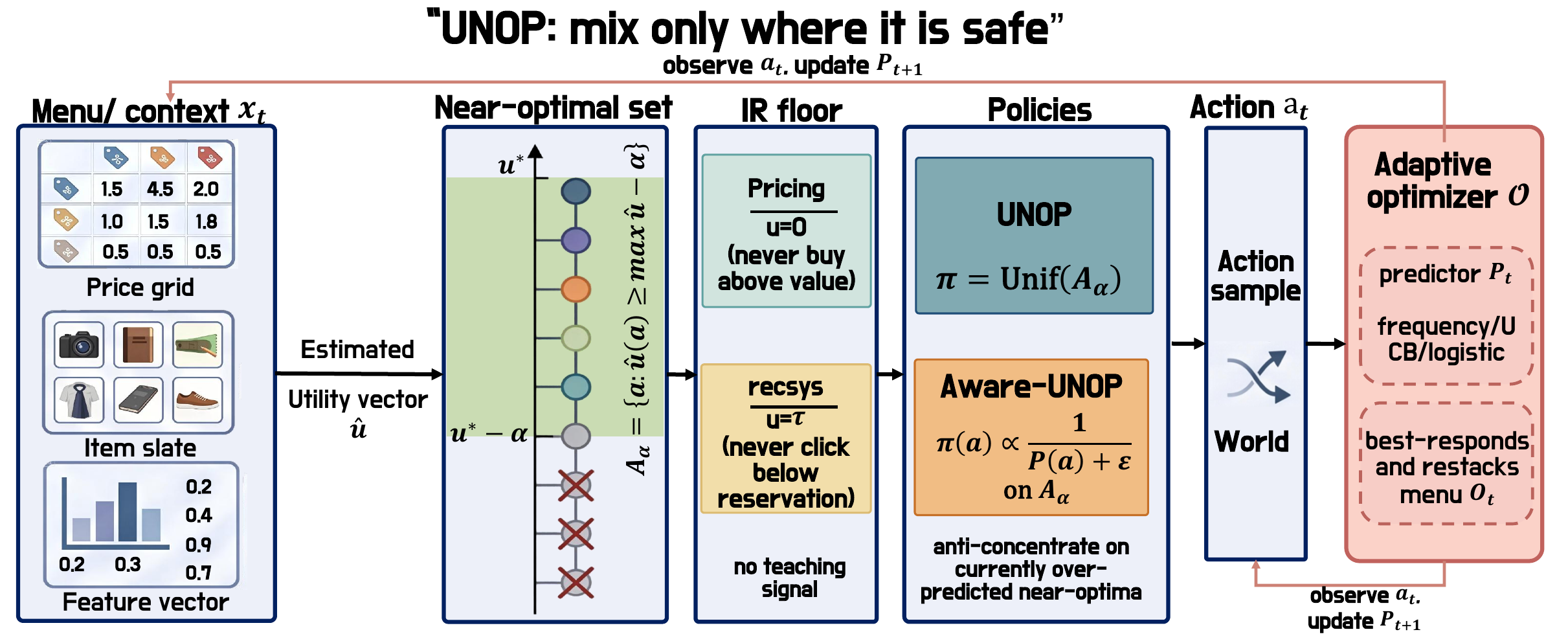}
\caption{UNOP pipeline. From the current menu the human forms $\hat u$, builds the $\alpha$-optimal set, intersects it with the reservation floor, and mixes uniformly on that set. The optimizer observes $a_t$, updates $P$, and revises the menu, and that update can converge to the mixture. The welfare effect comes from the best response to the mixture. In pricing, once the seller stays below the band, the relevant mixture is off the posted price.}
\label{fig:method}
\end{figure}

Unaware \unop{} is the unique solution of $(\mathrm{WC}_\alpha)$ on $\cA_\alpha^{\mathrm{IR}}$. When a forecast $P$ is available in the same simplex, \aunop{} replaces maximising $H(\pi)$ by a strictly concave objective that downweights actions the forecast already expects
\begin{equation}
\label{eq:aware}
\pi^{\mathrm{aware}}
\;\in\;
\arg\max_{\substack{\pi\in\Delta(A)\\ \mathrm{supp}(\pi)\subseteq\cA_\alpha^{\mathrm{IR}}}}
\Bigl( H(\pi) - \EE_{a\sim\pi}\bigl[\log\bigl(P(a)+\varepsilon\bigr)\bigr] \Bigr),
\end{equation}
with $\varepsilon>0$ a numerical floor. The unique maximizer is $\pi^{\mathrm{aware}}(a)\propto \bigl(P(a)+\varepsilon\bigr)^{-1}$ on $\cA_\alpha^{\mathrm{IR}}$, equivalently the $I$-projection
\begin{equation}
\label{eq:aware-kl}
\pi^{\mathrm{aware}}
\;=\;
\arg\min_{\pi\in\Delta(\cA_\alpha^{\mathrm{IR}})}\ \kl\bigl(\pi\,\big\|\,\nu_P\bigr),
\qquad
\nu_P(a)\;\propto\;\bigl(P(a)+\varepsilon\bigr)^{-1}.
\end{equation}
If $P$ is uniform on that set,~\eqref{eq:aware} reduces to unaware \unop{}. Algorithm~\ref{alg:unop} is the one-round map used in every experiment, and the surrounding loops are in Appendix~\ref{app:algs}. It is a good-enough mixer rather than a no-regret algorithm, and Hedge can be exploited by a patient optimizer~\citep{deng2019strategizing}. Across extraction domains, $\alpha$ on the order of $0.05$ to $0.15$ of the utility range is consistently near-best (Appendix Figure~\ref{fig:alpha}).

\section{Experiments}
\label{sec:exp}

\subsection{Experimental Settings}
\label{sec:setting}

Bimatrix games, posted pricing, and recommendation are the three interaction classes. The new comparisons below use an oracle seller that knows $q_\pi$, a Thompson seller, a stationary tax game, and a policy-aware slate optimizer. Seeds, horizons, and raw tables for the {UCB}, MovieLens, and tabular runs are in \texttt{supplement/legacy}. \aunop{} sees the true forecast $P_t$. Unaware \unop{} does not. Probabilities are clipped at $10^{-12}$ before every logarithm.

\subsection{Posted pricing}
\label{sec:results}
\label{sec:oracle}

\begin{figure}[t]
\centering
\includegraphics[width=0.78\linewidth]{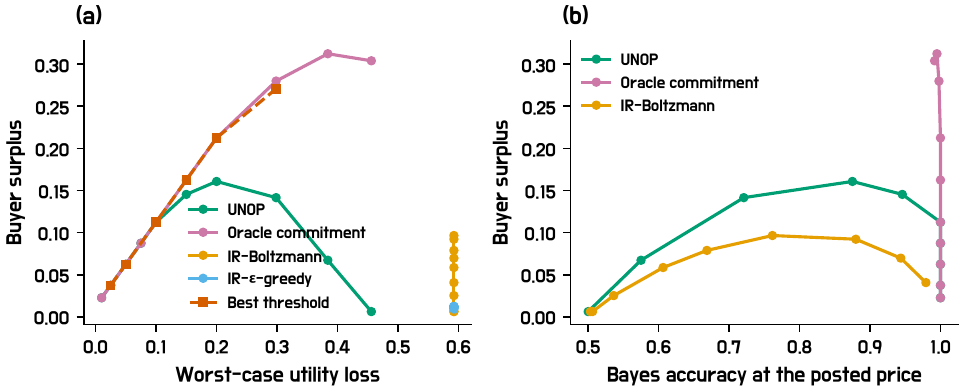}
\caption{Test-split oracle pricing. (a)~Buyer surplus against worst-case utility loss. \unop{} is the maximum-entropy support-safe point. The plotted {IR}-Boltzmann curve still leaves $\cA_\alpha$, while support-matched Boltzmann at the same $\alpha$ has surplus $\BoltMatch$, against $\UnopMatch$ for \unop{}. The oracle band commitment rises above \unop{} once $\rho=1/2$ fails~\eqref{eq:deter}. (b)~Bayes accuracy at the posted price can be high at the same time, because the mixture that keeps the seller below the band is off that price.}
\label{fig:frontier}
\end{figure}

We draw $200$ values $v\sim U[0.25,0.95]$ and hold out $120$ for test. The $\alpha$ we report is $0.15$, chosen before the test split is used. Mean-surplus selection on the validation split returns $0.20$, which raises the average and worsens the lower tail, so the test numbers stay at $0.15$.

The static test uses a demand oracle that knows $q_\pi$ and breaks revenue ties by lowest buyer surplus, then by highest price. On the $41$-point test grid the experimental unit is a value, greedy surplus is $\OracleGreedySurplus$, and \unop{} surplus is $\OracleUnopSurplus$. The paired difference is $\OracleUnopDiff$, with a $95\%$ bootstrap CI $[\OracleUnopLo,\OracleUnopHi]$ over $10{,}000$ resamples of the $\OracleN$ values, and the difference is positive on $\OracleUnopPos\%$ of them. A $10{,}001$-point average over values on $[0.25,0.95]$, with the $41$-point price grid held fixed, gives $\DenseGreedy$ and $\DenseUnop$. Bayes accuracy at the posted price is $\OracleUnopBayes$, because the purchase is sure whenever~\eqref{eq:deter} holds. The $10.8\%$ of test values on which \unop{} does not beat greedy are exactly the values on which $\rho=1/2$ fails~\eqref{eq:deter} and the seller enters the band. An oracle support-safe commitment, which may lower $\rho$ until~\eqref{eq:deter} holds, reaches mean surplus $\OracleWelfareSurplus$ on the same test split and $\DenseOracle$ under the same fine average. It uses the seller's revenue objective, so we treat it as an upper bound rather than a second rule a user could deploy. Support-restricted Boltzmann at the same $\alpha=0.15$ has surplus $\BoltMatch$, against $\UnopMatch$ for \unop{}, so the gain is carried by the hard support and the induced demand curve, with uniform mixing as the parameter-free maximum-entropy point of that class.

Table~\ref{tab:headline} places this static pricing comparison next to the online demand oracle, the two recommendation objectives, and the tax game. Green cells are the higher user payoff when the counterpart evaluates the response curve. The blue cell is quality ranking, where the list is not chosen against that curve and greedy remains ahead. The subsections below report these comparisons one setting at a time.

\begin{table}[t]
\centering
\caption{Headline user payoffs used in this section. Green is the higher payoff when the counterpart evaluates the response curve. Blue is quality ranking, where greedy stays ahead.}
\label{tab:headline}
\vspace{-0.4em}
{\small
\setlength{\tabcolsep}{8pt}
\renewcommand{\arraystretch}{1.05}
\begin{tabular}{lcc}
\hline
Comparison & Greedy & \unop{} \\
\hline
Static oracle pricing & $\OracleGreedySurplus$ & \cellcolor{cellgreen}$\OracleUnopSurplus$ \\
Online demand oracle & $\FamOracleG$ & \cellcolor{cellgreen}$\FamOracleU$ \\
Quality-ranking lists & \cellcolor{cellblue}$\RecVanillaGreedy$ & $\RecVanillaUnop$ \\
Policy-aware lists & $\RecExactGreedy$ & \cellcolor{cellgreen}$\RecExactUnop$ \\
Tax-game payoff & $\LockGreedyHuman$ & \cellcolor{cellgreen}$\LockHuman$ \\
\hline
\end{tabular}
}
\end{table}

Figure~\ref{fig:strength}a is a separate online comparison, with four sellers on $60$ random values for $300$ rounds, and those bars are not a ranked strength scale. Table~\ref{tab:sellers} records buyer surplus for greedy and \unop{} under each seller. Green marks where \unop{} pulls ahead. Under {UCB} the two policies stay close, because that seller is still walking down a $41$-point grid after $300$ rounds and has not yet evaluated the response curve. Under the demand oracle the per-price KL is $\FamOracleKL$.

\begin{table}[t]
\centering
\caption{Buyer surplus against four online sellers, $60$ values and $300$ rounds. Green marks a clear \unop{} gain. The {UCB} row stays close because that seller is still exploring.}
\label{tab:sellers}
\vspace{-0.4em}
{\small
\setlength{\tabcolsep}{8pt}
\renewcommand{\arraystretch}{1.05}
\begin{tabular}{lcc}
\hline
Seller & Greedy & \unop{} \\
\hline
{EXP3} & $\FamExpG$ & \cellcolor{cellgreen}$\FamExpU$ \\
{UCB} & $\FamUcbG$ & $\FamUcbU$ \\
Thompson sampling & $\FamThompG$ & \cellcolor{cellgreen}$\FamThompU$ \\
Demand oracle & $\FamOracleG$ & \cellcolor{cellgreen}$\FamOracleU$ \\
\hline
\end{tabular}
}
\end{table}

\begin{figure}[t]
\centering
\includegraphics[width=0.75\linewidth]{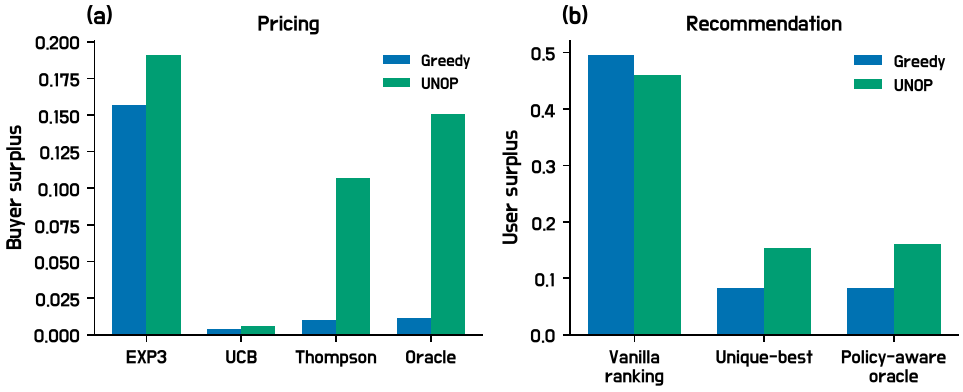}
\caption{(a)~Posted pricing under four online sellers, $60$ values and $300$ rounds, not a ranked strength scale. {EXP3} mixes exploration with importance-weighted updates, and {UCB} at this horizon is still exploring. Thompson sampling and the demand oracle are closer to a best response to $q_\pi$. (b)~Recommendation. The experimental unit is a catalog. Quality ranking favors greedy, and an exact policy-aware slate optimizer reverses the comparison.}
\label{fig:strength}
\end{figure}

\subsection{Recommendation}

On catalogs of $12$ items and slates of size $4$, every one of the $495$ slates is scored with the true $\pi(\cdot\mid S)$ over $160$ catalogs, and beam search of width $50$ matches that exact optimum. Quality ranking favors greedy, with paired difference $\RecVanDiff$ and $95\%$ CI $[\RecVanLo,\RecVanHi]$. The policy-aware oracle reverses the sign, so \unop{} minus greedy is $\RecExactDiff$, CI $[\RecExactLo,\RecExactHi]$, positive on $\RecExactPos\%$ of catalogs, with means $\RecExactUnop$ versus $\RecExactGreedy$. A separate $48$-item experiment, with slate size $8$, beam width $30$, and $80$ catalogs, gives paired difference $\RecBeamDiff$, CI $[\RecBeamLo,\RecBeamHi]$, positive on $\RecBeamPos\%$ of catalogs. A later latent-correlation sweep, using beam width $50$ and $60$ catalogs per cell, asks whether that gap depends on conflict between quality and margin (Appendix Figure~\ref{fig:conflict}). It does not trace support size. The policy-aware gap is largest under strong negative latent correlation, where at correlation $-0.8$ and reservation quantile $0.2$ the mean gap is $\ConflictNeg$, $95\%$ CI $[\ConflictNegLo,\ConflictNegHi]$. Under strong positive latent correlation the same gap is indistinguishable from zero ($\ConflictPos$, CI $[\ConflictPosLo,\ConflictPosHi]$). Several intermediate cells have a positive mean and a zero median, and quality ranking stays negative at every tested latent correlation from $-0.8$ to $0.8$. The sweep is about conflict, and it is not a map over how large the acceptable set is.

\subsection{Learning the mixture}

Figure~\ref{fig:learn} is a six-action stationary tax game with $100$ seeds. The user commits to a fixed mixture, the optimizer taxes the predicted mode, and the oracle taxes the true mode. For \unop{}, Bayes accuracy is $\LockBayes$, realized accuracy is $\LockAcc$, frequency KL is $\LockKL$, and optimizer regret is $\LockGap$, while user payoff is $\LockHuman$ against $\LockGreedyHuman$ for greedy. The predictor has reached the oracle best response, so what remains random is only the next draw from the learned mixture, and the user's payoff stays higher because that mixture changes which action is taxed.

\begin{figure}[t]
\centering
\includegraphics[width=0.75\linewidth]{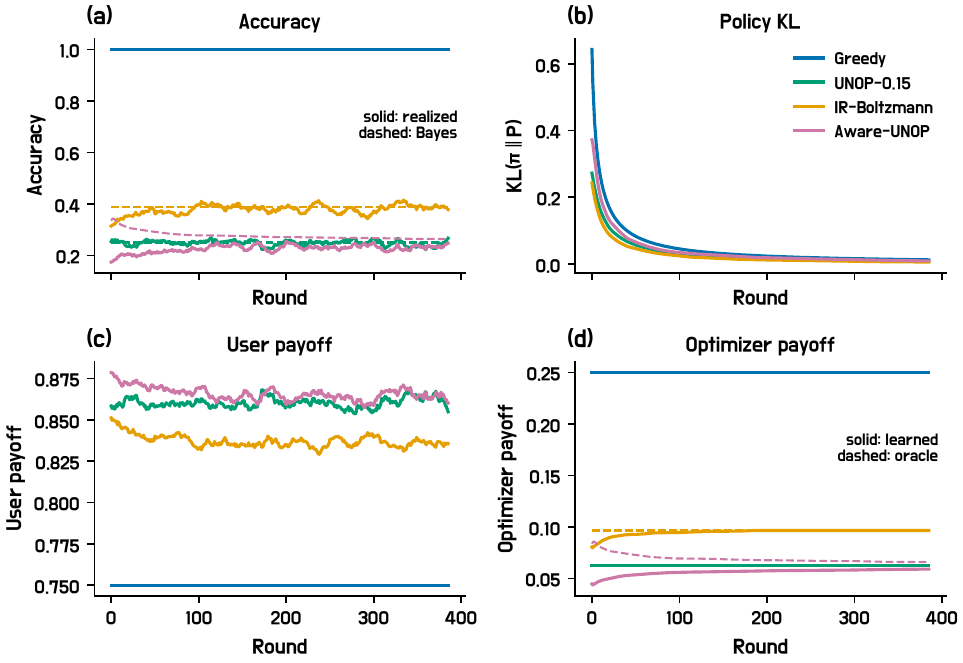}
\caption{Stationary tax game, frequency predictor, $100$ seeds, rolling means. Solid curves in (a) are realized accuracy and dashed curves are Bayes accuracy. Panel (b) is policy KL. In (d), solid curves are the learned optimizer and dashed curves are $V_O^\star(\pi)$, and their gap is optimizer regret rather than a measure of whether the policy can be exploited. The mixture is learned, the next draw stays at accuracy $1/k$, and user payoff stays above greedy.}
\label{fig:learn}
\end{figure}

\subsection{Discussion}

The pricing and recommendation gains come from the response curve the optimizer best-responds to. They are not evidence that the mixture stayed hidden. In the tax game the predictor's regret falls to zero, so the mixture has been identified, and the higher user payoff comes from which action that mixture makes costly to tax. The pricing misses line up with the same object. The held-out values on which uniform mixing does not beat greedy are the values on which a purchase rate of one half fails to keep the seller below the band, so those misses are named by the commitment condition rather than left as residual noise. Recommendation makes the contrast sharper. A ranker that ignores the response prefers the greedy list, while a slate optimizer that scores the true mixture reverses the comparison, and the reversal is largest when high quality and high margin pull apart. The short-horizon dynamic program is the boundary of the claim. Once the seller is still collecting samples instead of evaluating the curve, the optimal buyer collapses to greedy and the mixture can cost surplus.

Against a five-price {UCB} seller, exact dynamic programming at horizons $10$, $15$, $20$, and $25$ gives the optimal buyer and greedy the same value at every tested value, and \unop{} is weakly worse on several of those values, which fits a seller that has not yet evaluated the curve. A singleton acceptable set is a different failure. If $\cA_\alpha^{\mathrm{IR}}(S)$ is a singleton for every feasible menu $S$, \unop{} and greedy induce the same response, while a singleton only on the full candidate set still leaves the platform free to offer another menu. At $\alpha=0$ the two policies have identical buyer surplus on all $120$ test values, and the restricted MovieLens-100K cohort is a near-collapse, with mean difference $-0.0009$ over $33$ users. Utility noise $\hat u(a)=u(a)+\xi$, $\xi\sim\mathcal N(0,\sigma^2)$, against the oracle seller at $\alpha=0.15$ leaves surplus $\NoiseZero$ at $\sigma=0$ and $\NoiseThirty$ at $\sigma=0.30$, with a true loss-above-$\alpha$ rate of $\NoiseViol$. A practical reading is therefore narrow. Keep the mixture when the platform will best-respond to the whole curve and the menu still leaves more than one acceptable action, and drop it when the estimate of that set is wrong or the other side is a collaborator.

\section{Conclusion}
\label{sec:conc}

The menu moves because the platform best-responds to a response curve. \unop{} mixes uniformly on near-optimal, individually rational actions, and in pricing that mixture can keep a seller below the randomized band even after the curve is known. The gain appears against an extractive counterpart that evaluates the curve, and it disappears when the acceptable set collapses, the utility estimate is badly wrong, or the other side is still exploring.

\clearpage
\subsection*{Ethics statement}
This work studies how humans can avoid having their preferences precisely identified by adaptive pricing and recommendation systems. The methods can resist targeting that takes surplus and, in principle, useful personalization. We do not provide tools for attacking deployed systems. MovieLens-100K/1M~\citep{harper2015movielens}, {UCI} Adult~\citep{kohavi1996scaling}, German Credit~\citep{hofmann1994german}, FICO {HELOC}, {UCI} Credit Card, Give Me Some Credit, and {UCI} Bank Marketing are used as public research files. Any real-world assistant implementing \unop{} should make $\alpha$ user-visible and disable mixing in explicitly cooperative modes, as in Proposition~\ref{prop:common}.

\subsection*{Reproducibility statement}
Code and raw outputs for every reported comparison are at \url{https://github.com/FrankDengAI/unop} and in the supplement. Oracle pricing, the stationary tax game, dynamic programming, utility noise, the seller family, policy-aware recommendation, the grid condition, the fine value-grid average, the latent-correlation sweep, and the singleton-support runs are under \texttt{supplement/sprint}, while bimatrix, MovieLens, and tabular outputs are under \texttt{supplement/legacy}. MovieLens-100K \texttt{u.data} is included, and MovieLens-1M together with the OpenML tabular sources are fetched by \texttt{download\_data.py} and \texttt{run\_public.py} when absent. Values for the oracle suite use seed $20260926$, and the test split is the last $60\%$ of $200$ draws. Commands and package versions are in \texttt{supplement/README.md}. No GPU is required. Proofs are in Appendix~\ref{app:proofs}.

\subsection*{AI use statement}
Artificial intelligence tools assist in formatting and text polishing. The authors take full responsibility for the final content.

\bibliography{iclr2027_conference}
\bibliographystyle{iclr2027_conference}

\clearpage
\appendix
\setcounter{table}{0}
\setcounter{figure}{0}
\setcounter{algorithm}{0}
\renewcommand{\thetable}{A\arabic{table}}
\renewcommand{\thefigure}{A\arabic{figure}}
\renewcommand{\thealgorithm}{A\arabic{algorithm}}
\providecommand{\theHtable}{\thetable}
\providecommand{\theHfigure}{\thefigure}
\renewcommand{\theHtable}{A.\arabic{table}}
\renewcommand{\theHfigure}{A.\arabic{figure}}

\begin{center}
{\Large\bfseries Supplementary Material}\\[0.35em]
{\large Appendix}
\end{center}
\vspace{0.8em}

\noindent The appendices collect proofs of the claims in Section~\ref{sec:method}, the interaction loops used in every experiment, and further empirical detail. Algorithm~\ref{alg:unop} is the human's one-round mixer. Algorithms~\ref{alg:bimatrix} to \ref{alg:reclist} are the environments it is dropped into. They match the repository simulators on the objects that affect reported numbers: support construction, the {IR} intersection, {UCB} scores, and unique-best filler items.

\begin{algorithm}[H]
\caption{One-round \unop{} / \aunop{}}
\label{alg:unop}
{\small
\begin{algorithmic}[1]
\Require $u\in\R^{|A|}$, $\alpha\ge 0$, reservation $\underline{u}$, optional forecast $P$, aware flag, $\varepsilon>0$
\State $u^\star \leftarrow \max_a u(a)$ \quad $\cA \leftarrow \{a:u(a)\ge u^\star-\alpha\}$
\If{$\cA\cap\{a:u(a)\ge\underline{u}\}\neq\emptyset$}
    \State $\cA \leftarrow \cA\cap\{a:u(a)\ge\underline{u}\}$ \Comment{\eqref{eq:ir-set}}
\EndIf
\If{aware \textbf{and} $P$ is present}
    \State $\pi(a)\propto (P(a)+\varepsilon)^{-1}$ on $\cA$ \Comment{\eqref{eq:aware}}
\Else
    \State $\pi \leftarrow \mathrm{Unif}(\cA)$ \Comment{\eqref{eq:wc}. Unaware \unop{} ignores $P$}
\EndIf
\State \Return $a_t\sim\pi$
\end{algorithmic}
}
\end{algorithm}

\section{Earlier controlled comparisons}
\label{app:legacy}

\subsection{Comparison Results}
\label{sec:results-legacy}
\label{sec:exp-tri}
\label{sec:exp-price}
\label{sec:exp-rec}
Figure~\ref{fig:tri} plots, for each bimatrix game, predictive accuracy against human payoff normalized within the game. In zero-sum games the effective policies are the unpredictable ones: greedy and $\varepsilon$-greedy sit at the bottom, \aunop{} at the top. In common-interest games the mass concentrates at the top-right: being predicted is what good play looks like. In general-sum games the cloud is interior, and \unop{} occupies the attractive region, slightly above greedy in payoff, with lower accuracy. Absolute payoffs (Appendix~\ref{app:extra}) match the picture: zero-sum suite greedy $-0.143$ vs.\ \aunop{} $+0.054$. Coordination greedy $=$ \unop{} $=1.95$ vs.\ uniform $1.00$.

\begin{figure}
\centering
\includegraphics[width=0.8\linewidth]{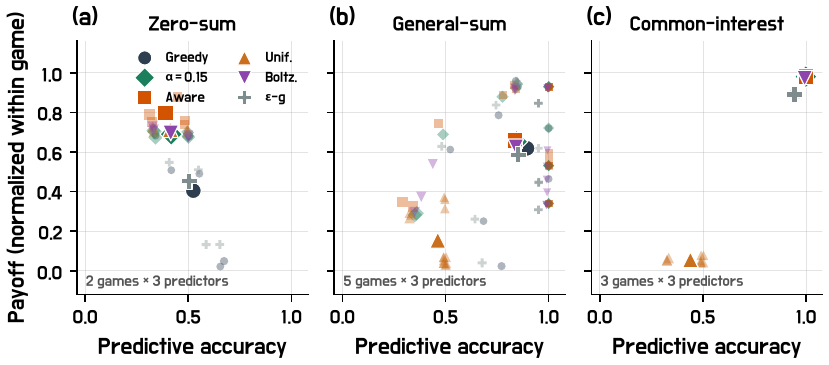}
\caption{(a)~Zero-sum: mixing is mandatory. (b)~General-sum: an interior frontier, and \unop{} sits on it. (c)~Common-interest: the Pareto point is pure. Faint markers are game and predictor pairs, with size proportional to policy entropy. Diamonds and circles are class-wise means $\pm$ sd.}
\label{fig:tri}
\end{figure}

Table~\ref{tab:main} and Figure~\ref{fig:price} summarize a {UCB} monopolist learning a hidden value $v\sim U[0.25,0.95]$. Greedy buyers are identified almost perfectly (accuracy $0.999$) and retain surplus $0.018$, on the order of the price grid. \unop{} with $\alpha=0.05$ more than doubles surplus to $0.042$ while remaining individually rational. $\alpha=0.15$ keeps $0.040$ at accuracy $0.807$ rather than $0.999$. \aunop{} matches this surplus at still lower accuracy ($0.785$). Boltzmann and uniform buy above value and finish negative. $\varepsilon$-greedy, which spends exploration on the dominated action, is worse than greedy. Table~\ref{tab:main-pred} reports the matching log-loss and optimizer-payoff numbers: under \unop{}-$\alpha=0.15$, pricing $V_O$ falls from $0.479$ (greedy) to $0.241$, and log loss rises from $0.026$ to $0.312$. Appendix Figure~\ref{fig:ts} shows the same runs as time series: greedy posted prices climb toward $v$. \unop{} keeps a lower posted price. The predictor's hit rate settles near the policy's Bayes accuracy rather than revealing a failure to learn the mixture (Section~\ref{sec:oracle}).

\begin{table}[H]
\caption{Main extraction results (mean over seeds). Surplus is buyer surplus in pricing and $q-\tau$ in recommendation. Synth.\ strat.\ utility is approval minus manipulation cost on the synthetic population (Adult/German are in Table~\ref{tab:public}). Green: best surplus in each extraction column, and the least-predictable effective recsys policy. Blue: extractor lock-in (greedy accuracy) or the weak-dilemma peak (greedy still wins on strategic utility).}
\label{tab:main}
\vspace{-0.4em}
{\small
\begin{center}
\begin{tabular}{lcccccc}
\hline
Policy & \multicolumn{2}{c}{Pricing} & \multicolumn{2}{c}{Recsys} & \multicolumn{2}{c}{Synth.\ strat.} \\
 & surplus & acc. & surplus & acc. & util. & acc. \\
\hline
Greedy & $0.018$ & \cellcolor{cellblue}$0.999$ & $0.052$ & \cellcolor{cellblue}$0.704$ & \cellcolor{cellblue}$0.517$ & $0.304$ \\
$\varepsilon$-greedy & $0.011$ & $0.954$ & $0.051$ & $0.489$ & $0.500$ & $0.285$ \\
Boltzmann $0.2$ & $-0.021$ & $0.665$ & $0.024$ & $0.169$ & $0.435$ & $0.237$ \\
Boltzmann $0.5$ & $-0.068$ & $0.595$ & $-0.043$ & $0.134$ & $0.402$ & $0.214$ \\
\unop{} $\alpha{=}0.05$ & \cellcolor{cellgreen}$\mathbf{0.042}$ & $0.962$ & $0.126$ & $0.315$ & $0.496$ & $0.253$ \\
\unop{} $\alpha{=}0.15$ & $0.040$ & $0.807$ & \cellcolor{cellgreen}$\mathbf{0.148}$ & $0.340$ & $0.458$ & $0.276$ \\
\unop{} $\alpha{=}0.3$ & $0.022$ & $0.735$ & $0.110$ & $0.435$ & $0.409$ & $0.220$ \\
\aunop{} $\alpha{=}0.15$ & $0.041$ & $0.785$ & $0.124$ & \cellcolor{cellgreen}$0.241$ & $0.442$ & $0.240$ \\
Uniform & $-0.133$ & $0.499$ & $-0.094$ & $0.107$ & $0.368$ & $0.199$ \\
\hline
\end{tabular}
\end{center}
}
\end{table}

\begin{table}[H]
\caption{Same runs as Table~\ref{tab:main}. Predictor log loss is higher when the policy is less readable, and optimizer payoff $V_O$ is lower when there is less extraction in pricing and recommendation. In the synthetic strategic setting, $V_O$ is institution accuracy.}
\label{tab:main-pred}
\vspace{-0.4em}
{\footnotesize
\begin{center}
\begin{tabular}{lcccccc}
\hline
Policy & \multicolumn{2}{c}{Pricing} & \multicolumn{2}{c}{Recsys} & \multicolumn{2}{c}{Synth.\ strat.} \\
 & nll & $V_O$ & nll & $V_O$ & nll & $V_O$ \\
\hline
Greedy & $0.026$ & $0.479$ & $0.847$ & $1.183$ & $1.596$ & $0.882$ \\
$\varepsilon$-greedy & $0.198$ & $0.446$ & $1.453$ & $1.129$ & $1.599$ & $0.880$ \\
Boltzmann $0.2$ & $0.593$ & $0.279$ & $2.639$ & $0.871$ & $1.589$ & $0.887$ \\
Boltzmann $0.5$ & $0.673$ & $0.346$ & $2.657$ & $0.951$ & $1.608$ & $0.888$ \\
\unop{} $\alpha{=}0.05$ & $0.082$ & $0.396$ & $1.694$ & $0.962$ & $1.589$ & $0.895$ \\
\unop{} $\alpha{=}0.15$ & $0.312$ & \cellcolor{cellgreen}$0.241$ & $2.057$ & \cellcolor{cellgreen}$0.623$ & $1.561$ & $0.902$ \\
\unop{} $\alpha{=}0.3$ & $0.408$ & $0.189$ & $1.873$ & $0.533$ & $1.587$ & $0.891$ \\
\aunop{} $\alpha{=}0.15$ & $0.309$ & $0.238$ & $2.080$ & $0.764$ & $1.598$ & $0.897$ \\
Uniform & $0.704$ & $0.462$ & $2.596$ & $1.009$ & $1.615$ & $0.887$ \\
\hline
\end{tabular}
\end{center}
}
\end{table}

\begin{figure}[t]
\centering
\vspace{-0.6em}
\includegraphics[width=0.8\linewidth]{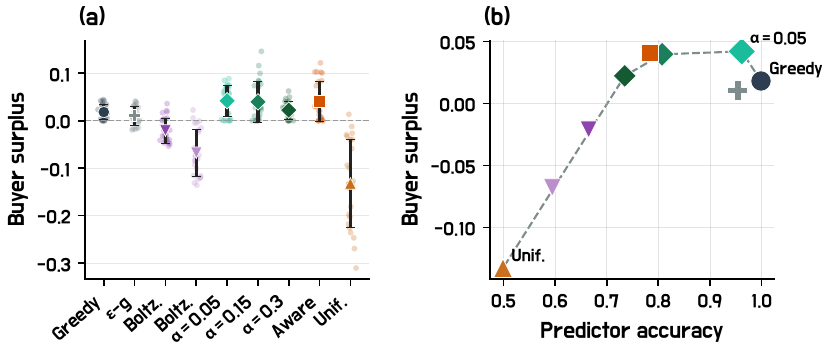}
\caption{(a)~Buyer surplus (seed strip + mean). (b)~The same runs in the effectiveness-predictability plane (marker size $=$ buy rate). Dashed path greedy $\to$ UNOP-$\alpha$ $\to$ uniform. Per-value scatter is in Appendix Figure~\ref{fig:A-pricing-scatter}. Boltzmann/uniform go negative because they buy above value.}
\label{fig:price}
\end{figure}

We first study a unique-best withholding extractor, and then a policy-aware slate optimizer that best-responds to each user policy (Section~\ref{sec:oracle}). In the unique-best design, once acceptance can be predicted, the platform withholds better organic content so that a high-margin item is the unique quality-best offer. Figure~\ref{fig:rec}(a) and Table~\ref{tab:main} show that greedy users are exactly the users this attack is designed for: surplus $0.052$, promoted-item rate $0.858$, accuracy $0.704$. \unop{} with $\alpha=0.15$ raises surplus to $0.148$ ($\approx 2.8\times$), cuts promoted exposure to $0.330$, and drops accuracy to $0.340$. \aunop{} is the least predictable effective policy (accuracy $0.241$). Uniform again goes negative. Removing the individual-rationality floor changes the sign of pricing surplus (Appendix Table~\ref{tab:A-ir}): \unop{}-$\alpha=0.15$ becomes negative once the policy may buy above value, because those purchases teach the monopolist that the current price is still acceptable.

\begin{figure}[t]
\centering
\vspace{-0.6em}
\includegraphics[width=0.8\linewidth]{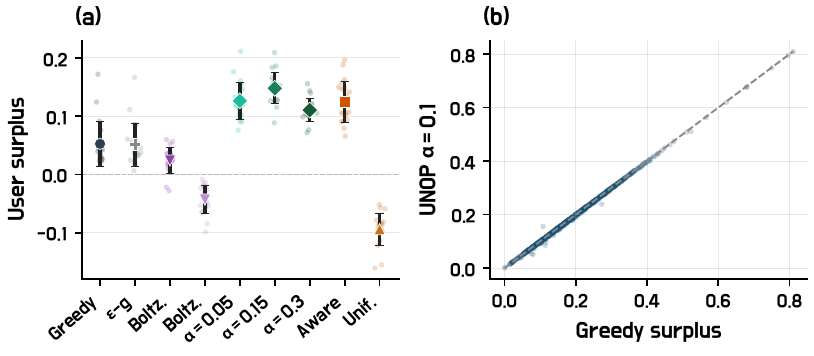}
\caption{(a)~Synthetic margin-seeking recommender: user surplus (seed strip + mean). (b)~Complete MovieLens-100K, per-user greedy vs.\ \unop{} scatter. Promoted-item rate versus surplus is in Appendix Figure~\ref{fig:A-recsys-scatter}.}
\label{fig:rec}
\end{figure}

\subsection{Discussion}
\label{sec:exp-more}
\label{sec:exp-weak}
The experiments do not say that unpredictability is always useful. They separate the cases in which a pure best response gives the optimizer a threshold, the cases in which there is no slack left to mix over, and the cases in which the counterpart is not trying to extract surplus at all. The practical question is where the $\alpha$-optimal set is large enough to mix over, and whether the other side gains by sitting on that threshold.

On MovieLens-100K (Figure~\ref{fig:rec}(b), all $943$ users and $6601$ user-policy jobs), five-star ratings often leave a single top item in the $\alpha$-optimal set, so \unop{} nearly reduces to greedy. Surplus is $0.189$ for greedy, $0.187$ for \unop{} at $\alpha=0.1$, and $-0.115$ for uniform. MovieLens-1M ($6{,}040$ users, $42{,}280$ jobs) gives the same ordering ($0.187$, $0.185$, and $-0.130$). Theorem~\ref{thm:frontiers} predicts this boundary. Coarse utilities remove the room to mix. Pricing and finely graded quality keep it.

A frequency predictor reaches accuracy $1.00$ on greedy behavior in $30$ steps. For \unop{} with $\alpha=0.15$, realized accuracy stays near $0.30$ through $350$ steps in $15$ of $16$ seeds (Appendix Figure~\ref{fig:lock}), which is the Bayes accuracy of a mixture on several near-optimal actions, not a failure to identify the mixture (Proposition~\ref{prop:decomp}). $\varepsilon$-greedy still concentrates, because $90\%$ of its mass stays on one action. Recommendation surplus peaks near $\alpha=0.10$ and then falls as skipping becomes excessive (Appendix Figure~\ref{fig:alpha}). With about $30\%$ \unop{} users, a shared margin-seeking optimizer's accuracy falls from $0.460$ on a purely greedy population to $0.219$, and greedy users' surplus rises from $0.143$ to about $0.17$ (Appendix Figure~\ref{fig:lock}, right). The optimizer no longer settles on one threshold, so part of the gain reaches users who do not change policy.

When the institution maximizes accuracy, lower predictability usually brings little or no utility gain. Across Adult and the other credit and marketing sets in Table~\ref{tab:public}, greedy stays ahead on utility. The counterpart's objective decides whether being readable is costly.

The same loop appears when an assistant learns which tasks, queues, or revisions a person will accept. That map helps delegation, and it can also steer later work toward the most readable person. Keeping several acceptable ways to carry out a task, when real slack exists, is the practical form of \unop{}. Proposition~\ref{prop:common} still applies when the counterpart is cooperative.

\section{Proofs}
\label{app:proofs}

The arguments below concern support, security level, coordination tax, demand curves, and Dirichlet concentration.

\begin{proof}[Proof of Proposition~\ref{prop:decomp}]
The hit rate of a predictor that names an action $a^\star$ is $\pi(a^\star)\le\max_a\pi(a)$, and the bound is achieved by naming a mode. For the log loss,
\[
\EE_{a\sim\pi}[-\log P(a)]
=-\sum_a\pi(a)\log\pi(a)+\sum_a\pi(a)\log\frac{\pi(a)}{P(a)}
=H(\pi)+\kl(\pi\|P).
\]
The KL of two distributions is zero if and only if they agree $\pi$-almost everywhere. Off the support the value of $P$ does not enter the expectation, which is why the statement is restricted to $\mathrm{supp}(\pi)$. Clipping at $10^{-12}$ is only a numerical guard in the experiments. The identity itself is stated for strictly positive probabilities on the support.
\end{proof}

\begin{proof}[Proof of Proposition~\ref{prop:zerosum}]
Write the human payoff matrix as $U$, with rows indexed by $a$ and columns by pure optimizer actions $o$. The security set is the nonempty polytope
\[
\mathcal S=\{\pi:\pi\ge0,\ \mathbf 1^\top\pi=1,\ U^\top\pi\ge v^\star\mathbf 1\}.
\]
If $\min_o u_H(\pi,o)\ge v^\star-\varepsilon$, then every violated security inequality has residual at most $\varepsilon$:
$\|[(v^\star\mathbf 1-U^\top\pi)_+]\|_\infty\le\varepsilon$.
Hoffman's error bound for a fixed feasible system of linear equalities and inequalities~\citep{hoffman1952approximate} gives a finite constant $C_G$ such that
$\inf_{q\in\mathcal S}\|\pi-q\|_1\le C_G\varepsilon$.
Since total variation is one half of the $\ell_1$ distance, the displayed claim follows with $\kappa_G=C_G/2$. This constant depends on the payoff matrix. The proposition does not assert a uniform constant over all games.

For a pure action $a$, an exact best responder realizes $\min_o u_H(a,o)$, which is strictly below $v^\star$ unless $a$ itself is a security strategy. In matching pennies with payoffs $\{+1,-1\}$, $\mathcal S=\{(1/2,1/2)\}$ and a policy $(1/2+d,1/2-d)$ has security value $-2|d|$ and total-variation distance $|d|$ from $\mathcal S$. Thus value within $\varepsilon$ requires $|d|\le\varepsilon/2$, ruling out a fixed peak as $\varepsilon\to0$. Rock, paper, scissors has the analogous unique uniform security strategy. This is the precise, game-dependent sense in which effective play in our zero-sum examples must be mixed.
\end{proof}

\begin{proof}[Proof of Proposition~\ref{prop:common}]
Let $u^\star=u(a^\star,o^\star)$ and $\bar u=\max_{a\ne a^\star,o}u(a,o)$. Uniqueness of the global maximizer implies $\bar u<u^\star$. The point mass on $a^\star$, paired with $o^\star$, attains $u^\star$. For any $\pi$ with $\pi(a^\star)=p$ and any fixed partner action $o$,
\[
\EE_{a\sim\pi}u(a,o)
\le p\,u^\star+(1-p)\bar u.
\]
The inequality therefore also holds after maximizing the left-hand side over the partner's actions. Relative to $u^\star$, every $p<1$ loses at least $(1-p)(u^\star-\bar u)$. Finally, grouping $a^\star$ against all other actions gives
$H(\pi)\le \mathrm h_2(p)+(1-p)\log(|A|-1)$. The right-hand side tends to zero as $p\to1$, so for each $h>0$ the constraint $H(\pi)\ge h$ implies $p\le p^\star(h)<1$.
\end{proof}

\begin{proof}[Proof of Theorem~\ref{thm:frontiers}]
(1) The feasible set of $(\mathrm{WC}_\alpha)$ is $\Delta(\cA_\alpha(u))$. Shannon entropy $H(\pi)=-\sum_a \pi(a)\log\pi(a)$ is strictly concave on this simplex and is uniquely maximised by the uniform distribution (Jensen). Thus $\mathrm{Unif}(\cA_\alpha)$ is the unique maximiser, $H=\log|\cA_\alpha|$, and $u(a)\ge u^\star-\alpha$ on that set implies $\EE_\pi[u]\ge u^\star-\alpha$. Moreover,
\[
\max_{a\in\mathrm{supp}(\pi)}(u^\star-u(a))\le\alpha
\quad\Longleftrightarrow\quad
\mathrm{supp}(\pi)\subseteq\cA_\alpha(u).
\]
Hence any positive mass outside $\cA_\alpha$ violates the advertised per-action guarantee. It need not violate an expected-loss constraint when that mass is small. $(\mathrm{WC}_\alpha)$ deliberately enforces the stronger support-wise condition.

(2) The Lagrangian for $\max_\pi \pi\cdot u + \tau H(\pi)$ over $\Delta(A)$ gives the Gibbs distribution~\citep{jaynes1957information,cover2006elements}: $\pi(a)\propto e^{u(a)/\tau}$. The map $\tau\mapsto H(\pi_\tau)$ is continuous and onto $(H(\mathrm{Unif}(\arg\max u)),\log|A|]$ for $\tau\in(0,\infty]$, so every $(\mathrm{SH}_h)$ with $h$ in that interval is attained at some $\tau>0$. Uniqueness follows from strict concavity of entropy. At $h=0$ the feasible set includes every mass on $\arg\max u$. The leak bounds~\eqref{eq:gibbs-leak} follow from $Z\le |A|e^{u^\star/\tau}$ and the existence of some $a\notin\cA_\alpha$ with $u(a)=u^\star-\alpha-\gamma$ for the lower bound, and from $Z\ge |\cA_\alpha|e^{(u^\star-\alpha)/\tau}$ together with $u(a)\le u^\star-\alpha-\gamma$ off $\cA_\alpha$ for the upper bound.
\end{proof}

\begin{figure}[t]
\centering
\includegraphics[width=0.82\linewidth]{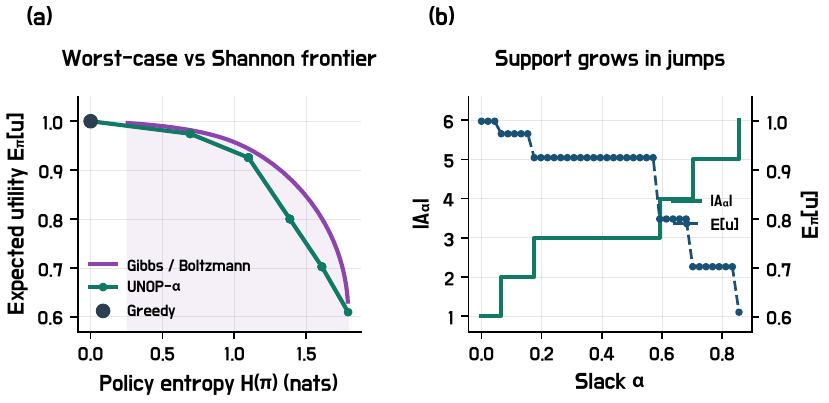}
\caption{(a)~Two Pareto curves on a six-action utility vector. Gibbs is optimal for a Shannon constraint and immediately puts weight outside the near-best set. \unop{} is optimal for a worst-case loss constraint. (b)~$|\cA_\alpha|$ jumps when the next action enters, while expected utility steps down.}
\label{fig:pareto}
\end{figure}

\begin{proof}[Proof of~\eqref{eq:aware}]
On $\Delta(\cA_\alpha^{\mathrm{IR}})$ the map $\pi\mapsto H(\pi)-\EE_\pi[\log(P+\varepsilon)]$ is strictly concave. The Lagrangian $\sum_a \pi(a)\log\frac{1}{\pi(a)(P(a)+\varepsilon)}+\lambda\bigl(1-\sum_a\pi(a)\bigr)$ has stationary point $\log\frac{1}{\pi(a)(P(a)+\varepsilon)}-1=\lambda$, hence $\pi(a)\propto (P(a)+\varepsilon)^{-1}$. The same point is the unique $I$-projection of $\nu_P\propto (P+\varepsilon)^{-1}$ onto the face, which is~\eqref{eq:aware-kl}.
\end{proof}

\begin{proof}[Proof of Theorem~\ref{thm:lockin}]
The posterior mean is $P_t(a)=(n_a+1)/(|A|+t)$. The strong law gives $n_a/t\to\pi(a)$ almost surely, hence $P_t\to\pi$. On $\mathrm{supp}(\pi)$ the map $P\mapsto\kl(\pi\|P)$ is continuous for $P$ bounded away from zero, which holds for large $t$ on that support, so the KL vanishes and the maximum coordinate converges to $\max_a\pi(a)$. For a $\delta$-peak on $a^\star$ with $\delta>1/2$, Hoeffding gives $\PP(n_{a^\star}\le t/2)\le\exp\bigl(-2t(\delta-1/2)^2\bigr)$. On that event $a^\star$ is the unique count maximizer. For uniform-on-$k$, $\EE[\max_a n_a]=t/k+\Theta\bigl(\sqrt{(t\log k)/k}\bigr)$ for fixed $k$, so
\[
\EE\bigl[\max_a P_t(a)\bigr]
=\frac{t}{k(|A|+t)}+\Theta\!\left(\frac{\sqrt{(t\log k)/k}}{|A|+t}\right),
\]
which tends to $1/k$.
\end{proof}

\begin{proof}[Proof of Proposition~\ref{prop:price}]
For $p<v-\alpha$, buy is the unique $\alpha$-optimal individually rational action, so $q=1$. For $p\in[v-\alpha,v]$, both skip and buy lie in the $\alpha$-set, and the constant-band rule sets $q=\rho$. For $p>v$, buy is individually irrational, so $q=0$. Sure-demand revenue is therefore maximized uniquely at $p^-_\alpha(v)$, with value $p^-_\alpha(v)$, whenever that price exists. Every band price earns at most $\rho\,p^+_\alpha(v)$. If $p^-_\alpha(v)>\rho\,p^+_\alpha(v)$, every band price, and every price above $v$, has strictly lower revenue, so the unique revenue maximizer is $p^-_\alpha(v)$. On a grid of spacing $\Delta$, $v-\alpha-\Delta<p^-_\alpha(v)<v-\alpha$, hence buyer surplus $v-p^-_\alpha(v)$ lies in $(\alpha,\alpha+\Delta]$. Any two rates that keep the seller below the band induce this same price and the same surplus. Let $r^\star=p^-_\alpha(v)/p^+_\alpha(v)$. Binary entropy on $(0,1)$ increases toward $1/2$. If $r^\star>1/2$, then $\rho=1/2$ lies in that set and uniquely maximizes entropy there. If $r^\star\le 1/2$, every such rate is strictly below $1/2$, entropy is increasing on $(0,r^\star)$, and the upper value $r^\star$ is not reached. A deployment that wants a maximizer can impose a revenue margin $\delta>0$, replacing~\eqref{eq:deter} by $p^-_\alpha(v)\ge \rho\,p^+_\alpha(v)+\delta$. The resulting set is closed and the entropy maximizer is $\min\bigl(1/2,(p^-_\alpha(v)-\delta)/p^+_\alpha(v)\bigr)$ whenever that quantity is positive.

If $v=2\alpha$ and $v\in\mathcal P$, then $p^+_\alpha(v)=v$ and band revenue at $\rho=1/2$ is $v/2=\alpha$, while every sure-demand price is strictly below $v-\alpha=\alpha$. Condition~\eqref{eq:deter} fails. The seller can prefer the band, including $p=v$, where buyer surplus is zero. Thus $\alpha\le v/2$ does not imply that the seller stays below the band under seller-favorable tie-breaking.
\end{proof}

\begin{figure}[t]
\centering
\includegraphics[width=0.62\linewidth]{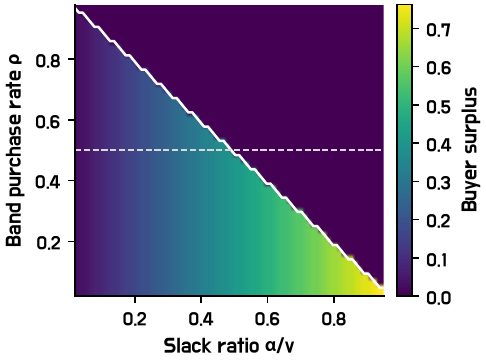}
\caption{Buyer surplus for a constant-band response on an $81$-point grid at $v=0.8$. The white curve is the open boundary $\rho=r^\star$. Seller-favorable tie-breaking does not treat the boundary itself as keeping the seller out. The dashed line is $\rho=1/2$. To the right of their intersection, uniform mixing fails~\eqref{eq:deter}, and the highest entropy of the rates that keep the seller below the band is not reached.}
\label{fig:phase}
\end{figure}

\begin{figure}[t]
\centering
\includegraphics[width=0.55\linewidth]{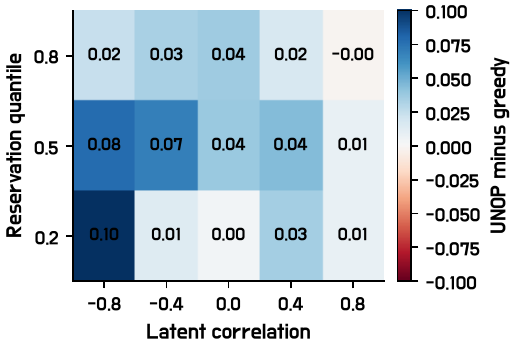}
\caption{Policy-aware beam search of width $50$ on catalogs drawn with a chosen Gaussian correlation, $n=48$, slate size $8$, $\alpha=0.15$, $60$ catalogs per cell. On $450$ matching $12$-item catalogs ($30$ per cell, both policies, $900$ policy and catalog rows) the same width matches exact enumeration in every case. Entries are mean user-surplus differences, \unop{} minus greedy. The horizontal axis is the correlation used to draw the pair, which is not the sample Pearson correlation of the realized quality and margin. The gain is largest at strong negative correlation. At strong positive correlation the paired interval covers zero. Paired intervals, medians, and positive fractions are in \texttt{supplement/sprint/results/recsys\_phase\_ci.json}. Quality ranking is negative in every cell.}
\label{fig:conflict}
\end{figure}

\section{Simulator algorithms}
\label{app:algs}

Algorithm~\ref{alg:unop} is the human's one-round map. Everything that surrounds it in the experiments is an adaptive optimizer that (i)~posts a menu, (ii)~observes $a_t$, (iii)~updates a predictor or a bandit count, and (iv)~posts a new menu. The three loops below are those optimizers. They match \texttt{experiments/simulate.py}. They specify the individual-rationality floor, the $5\%$ random slip, and the unique-best filler-item construction.

Two implementation details that Algorithm~\ref{alg:unop} compresses. First, under exact arithmetic the {IR} intersection is nonempty whenever the reservation set is nonempty, so~\eqref{eq:ir-set} reduces to intersecting $\cA_\alpha$ with $\{a:u(a)\ge\underline{u}\}$. If floating-point noise empties that set, the code falls back to $\cA_\alpha$. Second, \aunop{} uses $\varepsilon=10^{-6}$ in the $(P+\varepsilon)^{-1}$ weights. Both choices are implementation details. They are what make the released rows match the algorithms below.

What is being predicted.
In bimatrix games the predictor $P_t$ is a distribution on the human's action set, and the optimizer chooses an exact best response to that prediction. In posted pricing there is no recommendation list: the menu is the posted price, and $P_t$ is a Bernoulli on $\{\mathrm{skip},\mathrm{buy}\}$ at that price, maintained as a Beta posterior. In recommendation $P_t$ is an empirical distribution on the current recommendation list plus skip. Lock-in experiments use the frequency predictor of Algorithm~\ref{alg:bimatrix} on a synthetic six-action menu and declare $t_{\mathrm{lock}}$ at the first time a length-$30$ rolling accuracy exceeds $0.7$.

\begin{algorithm}[htbp]
\caption{Bimatrix closed loop (frequency / Markov / softmax predictor)}
\label{alg:bimatrix}
{\small
\begin{algorithmic}[1]
\Require game $(U_H,U_O)$, human policy $\pi$, horizon $T$, predictor $P$
\State $\hat\sigma \leftarrow \mathrm{Unif}(O)$ \Comment{empirical mix of optimizer actions}
\For{$t=1,\ldots,T$}
    \State $\hat u \leftarrow U_H\hat\sigma$ \Comment{mean-based human,~\eqref{eq:mean-based}}
    \State $a_t \sim \pi(\cdot\mid \hat u, P_t)$ \Comment{Algorithm~\ref{alg:unop}, or a baseline}
    \State $o_t \leftarrow \arg\max_o (U_O^\top P_t)_o$, with a $5\%$ uniform random slip
    \State observe payoffs $U_H(a_t,o_t)$, $U_O(a_t,o_t)$
    \State update $P_{t+1}$ from $a_t$. \quad $\hat\sigma \leftarrow$ empirical mix of $o_{1:t}$
\EndFor
\State \Return post-burn-in means of $u_H$, $u_O$, $\mathbf{1}\{\arg\max P_t = a_t\}$, and $H(\pi_t)$
\end{algorithmic}
}
\end{algorithm}

The $5\%$ random slip is there so that a pure profile cannot freeze the dynamics by accident. It is also why greedy in matching pennies finishes near $-0.14$ rather than at $-1$: identification is fast, not instantaneous, and the opponent occasionally plays off the punishing action. Reported bimatrix means drop the first quarter of the horizon.

\begin{algorithm}[htbp]
\caption{{UCB} posted-price monopolist}
\label{alg:pricing}
{\small
\begin{algorithmic}[1]
\Require price grid $\{p_j\}_{j=1}^{m}$, hidden value $v$, buyer policy $\pi$, horizon $T$
\State $s_j \leftarrow 1$, $f_j \leftarrow 1$ for all $j$ \Comment{Beta-Bernoulli prior on $\PP(\mathrm{buy}\mid p_j)$}
\For{$t=1,\ldots,T$}
    \State $\hat b_j \leftarrow s_j/(s_j+f_j)$
    \State $\mathrm{UCB}_j \leftarrow p_j\cdot\mathrm{clip}\bigl(\hat b_j+\sqrt{2\log(t+2)/(s_j+f_j)},\,0,1\bigr)$
    \State post $p_t \leftarrow p_{j^\star}$ with $j^\star=\arg\max_j \mathrm{UCB}_j$
    \State $\hat b_{j^\star} \leftarrow s_{j^\star}/(s_{j^\star}+f_{j^\star})$. \quad $P_t \leftarrow (1-\hat b_{j^\star},\,\hat b_{j^\star})$
    \State buyer utilities $u\leftarrow(0,\,v-p_t)$. \quad $a_t\sim\pi(\cdot\mid u, P_t)$ with reservation $0$
    \If{$a_t=\mathrm{buy}$}
        \State $s_{j^\star}\leftarrow s_{j^\star}+1$
    \Else
        \State $f_{j^\star}\leftarrow f_{j^\star}+1$
    \EndIf
\EndFor
\State \Return post-burn-in buyer surplus, posted price, buy rate, predictor accuracy
\end{algorithmic}
}
\end{algorithm}

Algorithm~\ref{alg:pricing} instantiates Proposition~\ref{prop:price}. A greedy buyer produces a step-function demand curve, so the {UCB} scores concentrate on the last grid cell below $v$ and surplus is about one grid step. An {IR}-\unop{} buyer produces the three-piece demand~\eqref{eq:unop-demand}, so the same scores cannot treat the coin-flip band as sure demand. Thompson sampling replaces the {UCB} bonus by a Beta draw (Appendix~\ref{app:optvar}). The qualitative gap survives. Turning the reservation off is Appendix~\ref{app:ir}: the buyer then sometimes buys above value, the success count $s_j$ is incremented at $p>v$, and surplus changes sign.

\begin{algorithm}[htbp]
\caption{Margin-seeking unique-best recommendation list}
\label{alg:reclist}
{\small
\begin{algorithmic}[1]
\Require catalog $(q_i,m_i)_{i=1}^{n}$, list size $k$, reservation $\tau$, user policy $\pi$, horizon $T$
\State $N_i\leftarrow 1$, $C_i\leftarrow 1$ for all $i$ \Comment{counts when $i$ was the unique-best offer}
\For{$t=1,\ldots,T$}
    \State $\hat p_i \leftarrow C_i/N_i$. \quad $\mathrm{UCB}_i \leftarrow \sqrt{2\log(t+2)/N_i}$
    \State $\mathrm{profit}_i \leftarrow m_i\cdot\mathrm{clip}(\hat p_i + 0.7\cdot\mathrm{UCB}_i,\,0,1.2)$
    \State offer $i^\star \leftarrow \arg\max_i \mathrm{profit}_i$ (uniform among items for the first $30$ rounds, and with an $8\%$ random slip thereafter)
    \State fill the remaining $k-1$ slots with items of quality strictly below $q_{i^\star}$
    \State $P_t \leftarrow$ predictive distribution on the displayed actions (simulation: extractor forecast on $i^\star$, remaining mass on skip and filler items)
    \State user utilities $\leftarrow (q_{j_1},\ldots,q_{j_k},\tau)$. \quad $a_t\sim\pi(\cdot\mid u, P_t)$ with reservation $\tau$
    \State $N_{i^\star}\leftarrow N_{i^\star}+1$
    \If{the click is $i^\star$}
        \State $C_{i^\star}\leftarrow C_{i^\star}+1$
    \EndIf
    \State surplus is $q_{\mathrm{clicked}}-\tau$ on a click and $0$ on skip
\EndFor
\end{algorithmic}
}
\end{algorithm}

The unique-best construction is the extraction. A vanilla ranking {UCB} that also shows better organic items is a weaker adversary, because a greedy quality maximizer then clicks the organic item and never certifies the high-margin offer. MovieLens uses the same loop with $q_{ui}$ from the official rating files, scaled to $[0,1]$, and the same quantile reservation. When the five-star grid makes $\cA_\alpha$ a singleton at the top star, Algorithm~\ref{alg:unop} returns greedy, which is the boundary in Theorem~\ref{thm:frontiers}.

Baselines, for the same menus.
Greedy is Algorithm~\ref{alg:unop} at $\alpha=0$. $\varepsilon$-greedy puts mass $\varepsilon/|A|$ on every action, including those below reservation. Boltzmann is $\pi(a)\propto\exp(u(a)/\tau)$ on the whole of $A$. Uniform is $1/|A|$. Hedge is multiplicative weights on realized $\hat u_t$, with no support constraint. Under Algorithms~\ref{alg:pricing} and~\ref{alg:reclist} the extractor does not change: only the human's support does. Gibbs leak (Theorem~\ref{thm:frontiers}) is visible in Algorithm~\ref{alg:pricing} as soon as Boltzmann buys at $p>v$ and increments the success count $s_j$ there.

Lock-in.
On a six-action synthetic menu with a planted $\delta$-peak or a uniform $k$-set, a frequency predictor $P_t(a)=(n_a+1)/(|A|+t)$ is updated every round. $t_{\mathrm{lock}}$ is the first $t$ at which a trailing window of $30$ one-hot hits $\mathbf{1}\{\arg\max P_s=a_s\}$ has mean at least $0.7$. Theorem~\ref{thm:lockin} is the concentration statement behind that curve: a $\delta$-peak with $\delta>1/2$ produces a unique posterior mode in $O(\log(1/\eta)/(\delta-1/2)^2)$ samples, while a uniform mix on $k$ near-ties keeps the posterior an $O(1/k)$-peak. $\varepsilon$-greedy still locks quickly because $1-\varepsilon$ of its mass is a peak. Entropy of $\cA_\alpha$, not a generic exploration rate, is what delays identification.

\section{Additional experimental details}
\label{app:extra}

All simulators are {NumPy}/{SciPy}/{scikit-learn} {CPU} code. We used $96$ processes. Controlled suite ($3240$ bimatrix jobs, $216$ pricing, $180$ recsys, $144$ synthetic strategic, $144$ lock-in, $452$ $\alpha$-sweep rows, $132$ population runs, $4508$ recorded rows) completed with zero errors. Public source files: MovieLens-100K, all $943$ users $\times$ $7$ policies ($6601$ jobs, official $\ge 20$ ratings). MovieLens-1M, all $6{,}040$ users $\times$ $7$ ($42{,}280$ jobs). {UCI} Adult, source population $48{,}842$, $d=108$ after one-hot, $16$ seeds $\times$ $9$ policies ($T=3000$, pretrained logistic accuracy $0.852$). German Credit, source population $1{,}000$, $d=61$, $16$ seeds $\times$ $9$ ($T=2000$, pretrained accuracy $0.785$). FICO {HELOC}, source population $10{,}000$, $16$ seeds $\times$ $9$. {UCI} Credit Card, source population $30{,}000$, $16$ seeds $\times$ $9$. Give Me Some Credit, source population $150{,}000$, pretrained accuracy $0.934$, $16$ seeds $\times$ $9$. {UCI} Bank Marketing, source population $45{,}211$, $d=51$, pretrained accuracy $0.901$, $16$ seeds $\times$ $9$. Bimatrix: $T=1000$, $12$ seeds, predictors $\{\mathrm{freq},\mathrm{Markov},\mathrm{logreg}\}$. Pricing: $T=500$, $24$ seeds, $21$-price grid. Recsys: $T=300$, $20$ seeds, $48$ items, list size $8$. Synthetic strategic: $T=700$, $16$ seeds, $200$ agents, $8$ features, $5$ manipulation levels. Lock-in: $T=350$, $6$ actions, $16$ seeds. Public files are loaded without file-level subsampling. MovieLens covers every user, while tabular interactions use the stated horizons.

Table~\ref{tab:public} reports greedy vs.\ \unop{} on every official public file we ran. MovieLens columns use $\alpha=0.1$. Tabular strategic-classification rows use $\alpha=0.15$.

\begin{table}[htbp]
\centering
\caption{Public complete datasets. $N$ is users for MovieLens and records for the tabular sets. Utility is user surplus on MovieLens and agent payoff, approval minus gaming cost, on the strategic-classification sets. MovieLens uses $\alpha=0.1$ because the ratings have almost no slack. Tabular rows use $\alpha=0.15$, where greedy is still ahead on utility and \unop{} is less readable. The Adult $\alpha=0.05$ sensitivity is utility/accuracy $0.181$/$0.428$, reported in Appendix~\ref{app:adult}.}
\label{tab:public}
\begin{center}
\small
\setlength{\tabcolsep}{4.5pt}
\renewcommand{\arraystretch}{1.12}
\begin{tabular}{lccccc}
\hline
Dataset & $N$ & greedy util. & greedy acc. & UNOP util. & UNOP acc. \\
\hline
MovieLens-100K (complete) & 943 & 0.189 & 0.415 & 0.187 & 0.418 \\
MovieLens-1M (complete) & 6040 & 0.187 & 0.419 & 0.185 & 0.425 \\
UCI Adult (complete) & 48842 & \cellcolor{cellblue}0.178 & 0.482 & 0.136 & 0.295 \\
German Credit (complete) & 1000 & \cellcolor{cellblue}0.757 & 0.618 & 0.722 & 0.328 \\
FICO HELOC & 10000 & \cellcolor{cellblue}0.501 & 0.779 & 0.470 & 0.312 \\
UCI Credit Card & 30000 & \cellcolor{cellblue}0.082 & 0.981 & 0.046 & 0.334 \\
Give Me Some Credit & 150000 & \cellcolor{cellblue}0.005 & 0.938 & -0.031 & 0.333 \\
UCI Bank Marketing & 45211 & \cellcolor{cellblue}0.061 & 0.815 & 0.024 & 0.327 \\
\hline
\end{tabular}
\end{center}

\end{table}

\begin{figure}[t]
\centering
\includegraphics[width=0.88\linewidth]{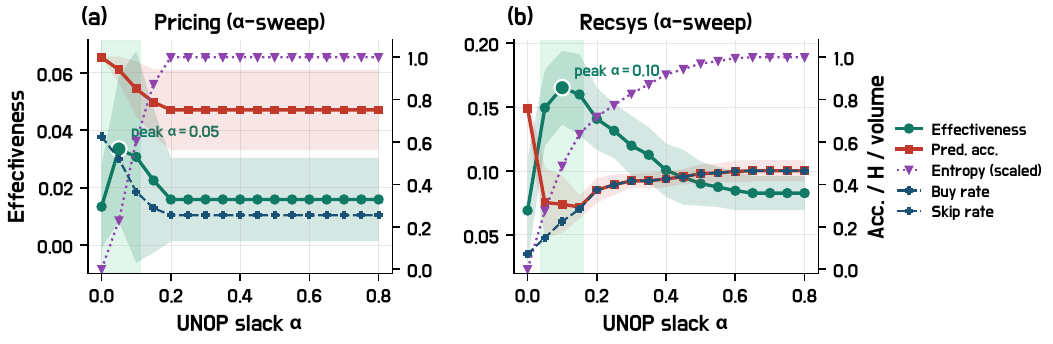}
\caption{(a)~Pricing $\alpha$-sweep. (b)~Recommendation $\alpha$-sweep. Shared legend at right: effectiveness rises and then falls. Buy/skip volume and scaled entropy show how the mix changes. Predictability falls once the $\alpha$-set is non-singleton.}
\label{fig:alpha}
\end{figure}

Matching pennies: greedy $-0.146$, \unop{}-$\alpha=0.15$ $0.000$, \aunop{} $+0.055$. {RPS}: greedy $-0.141$, \unop{} $0.00$, \aunop{} $+0.053$. Zero-sum suite means remain $-0.143$ / $0.000$ / $+0.054$. Coordination: greedy $=$ \unop{} $=1.95$, uniform $1.00$. Stag hunt: greedy $=$ \unop{} $=4.88$, uniform $3.96$. Chicken: \unop{}-$\alpha=0.15$ $1.82$ vs.\ greedy $1.77$. Inspection: \unop{}-$\alpha=0.05$ $0.96$ vs.\ greedy $0.92$. Prisoner's dilemma: all near-optimal policies collapse to defection against an adaptive counterpart ($1.11$), as expected.

In the margin-seeking recommender, greedy users click a promoted item in $85.8\%$ of rounds after the early exploration. \unop{}-$\alpha=0.15$ in $33.0\%$. \aunop{} in $41.5\%$. Skip rates rise with $\alpha$ ($0.064$ greedy, $0.340$ at $\alpha=0.15$), which is the volume and surplus tradeoff in Figure~\ref{fig:alpha}.

If the human's $\hat u$ is $u$ plus independent Gaussian noise of scale $\sigma$, the $\alpha$-set is a noisy near-best set. For $\sigma\le \alpha/2$, \unop{} remains inside the true $(2\alpha)$-optimal set with high probability on small action spaces, which is why a slightly larger $\alpha$ is a robust default.

Bimatrix results average frequency, Markov, and online-softmax predictors. The three-case pattern is stable across the three. Replacing margin-seeking unique-best recommendation lists with vanilla {CTR}-{UCB} ranking weakens extraction: without withholding, greedy quality-maximizers are not forced onto the margin.

\subsection{Restatement of the claims}
\label{app:theory-long}

This section restates the claims of Appendix~\ref{app:proofs} in the notation of the main text. The later ablations are read against that payoff shape.

Start with a utility vector $u\in\R^{|A|}$ that the human actually believes, or at least acts as if they believe. The number $u^\star=\max_a u(a)$ is the one-step best-response value, and the set $\cA_\alpha(u)=\{a:u(a)\ge u^\star-\alpha\}$ is the set of actions that are never more than $\alpha$ worse than that value. Theorem~\ref{thm:frontiers} separates two constraints. First, if every action that may be played lies in $\cA_\alpha$, then among all such policies the uniform distribution is the unique entropy maximizer, and every such policy has expected utility at least $u^\star-\alpha$ no matter how the rest of the vector is filled in. Second, if the only constraint is that entropy be at least some $h$, then the unique expected-utility maximizer is a Gibbs distribution, and Gibbs distributions put positive mass on every action as soon as the temperature is positive. The first constraint is a worst-case loss constraint and the second is a Shannon constraint. They coincide only when the utility vector is essentially flat on a subset and $-\infty$ elsewhere, which is not the situation a buyer or a recommender user is in.

That distinction is not a taste in regularizers. In posted pricing, a Gibbs buyer with even a modest temperature will sometimes buy above value, because the softmax never quite hits zero on the dominated action, and a learning monopolist treats those purchases as evidence that the posted price is still acceptable. \unop{} with an individual-rationality floor of zero will not do that. The ablation in Appendix~\ref{app:ir} is the empirical counterpart of this paragraph: when we turn the floor off, \unop{} starts to look like a slightly better-behaved Boltzmann policy, and the monopolist extracts it.

Proposition~\ref{prop:zerosum} is the matching-pennies version of the same payoff shape. If the opponent can identify a pure action and then choose the payoff-maximizing reply, then any policy that is a $\delta$-peak for large $\delta$ is giving away the security level. Mixing on the near-optimal set is effective play in this case, and \aunop{} is a small refinement that puts less weight where the current forecast is already high, instead of mixing blindly. Proposition~\ref{prop:common} is the opposite case. If the partner is trying to coordinate, a peak attains the Pareto-optimal payoff and mixing is a tax with no return from surplus-taking, so mixing is useful when there is room among near-best actions and the other side wants to take surplus.

Proposition~\ref{prop:price} is the claim that most of the empirical work is about. A greedy buyer publishes a step-function demand curve, so any revenue maximizer that knows that curve posts in the last grid cell at or below $v$ and buyer surplus is then on the order of one grid step. An individually rational \unop{} buyer publishes a band, with sure demand below $v-\alpha$, purchase rate $1/2$ on $[v-\alpha,v]$, and sure refusal above $v$. Let $p^-=\max\{p\in\mathcal P:p<v-\alpha\}$ and $p^+=\max\{p\in\mathcal P:v-\alpha\le p\le v\}$. The sure-demand price strictly keeps the seller below every price in that half-purchase band when $p^->p^+/2$. On a grid of spacing $\Delta$, $\alpha<v/2-\Delta$ is a convenient sufficient condition for this inequality, and the buyer surplus is then $\Theta(\alpha)$ up to one grid step. Otherwise the proposition claims the demand curve, and a lower band rate is needed before the same surplus bound holds. Theorem~\ref{thm:lockin} supplies the related concentration intuition, since a peak is identified in $O(\log(1/\eta))$ plays, whereas a uniform mixture on $k$ near-ties leaves only an $O(1/k)$ posterior peak.

None of this requires the human to solve a long-horizon Stackelberg problem against a mean-based learner~\citep{deng2019strategizing}, because the policies we study are one-step mixers. A fully patient optimizer may still extract them. The experiments instead isolate the behavior of the named {UCB} and recommendation-list extractors.

\subsection{Experimental protocol}
\label{app:protocol}

This section records implementation choices that the main experiments section only summarizes. Every number in the paper comes from the released result files. The only additional runs are the ablations in Appendices~\ref{app:ir} to \ref{app:optvar}.

What the human sees.
In bimatrix games the human does not see the optimizer's current mixed action as a one-hot. They see the empirical mix of optimizer actions so far, form $\hat u = U_H \hat\sigma$, and then apply $\pi(\cdot\mid \hat u, P)$. That is a mean-based, not a fully Bayesian, human, and it is closer to how a person looking at a platform's recent prices actually has to reason. In pricing the menu is a single posted price, so $\hat u=(\,0,\,v-p\,)$ is exact up to the noise ablation. In recommendation the menu is the current recommendation list plus the skip action, with skip's utility equal to the reservation $\tau$. We set $\tau$ to the $0.42$ quantile of item quality in the synthetic catalog so that the outside option lies inside the quality range. On MovieLens we keep the same construction after scaling ratings from $\{1,2,3,4,5\}$ into $[0,1]$.

What the optimizer sees.
Bimatrix predictors are frequency counts, a first-order Markov chain, or online softmax regression on the optimizer's own mix. The optimizer then chooses an exact best response to $P_t$, with a $5\%$ uniform random slip so that the dynamics cannot freeze on a pure profile by accident. The pricing optimizer is per-price Beta-Bernoulli {UCB} on $p\cdot\PP(\mathrm{buy}\mid p)$. The recommender is a unique-best margin-seeking recommendation list: one offer is the quality-best item on the recommendation list, the rest are strictly worse filler items, and {UCB} is run on the offer's acceptance. A vanilla ranking {UCB} that does not withhold better organic content is a weaker extractor, because it does not sit on $\tau$.

Burn-in.
Reported means drop the first quarter of the horizon in bimatrix, recommendation, and strategic settings, and the first third in pricing, because we are interested in what an adaptive system does after it has had a chance to look at the human, not in the first few exploratory prices. Lock-in curves are the exception: they are the object of interest, so we keep the whole trajectory and define $t_{\mathrm{lock}}$ as the first time a length-$30$ rolling accuracy reaches at least $0.7$.

Seeds and pooling.
Bimatrix: $12$ seeds $\times$ $3$ predictors $\times$ $9$ policies $\times$ $10$ games. Pricing: $24$ values of $v$. Recsys: $20$ catalog draws. Strategic: $16$ population draws. Public tables: MovieLens uses one job per user per policy on the official file. Adult, German, FICO {HELOC}, {UCI} Credit Card, Give Me Some Credit, and {UCI} Bank Marketing use $16$ seeds with a shared pretrained logistic classifier, then online updates on manipulated features. Ablations use $8$ to $24$ seeds as listed in each subsection. Tables report means over seeds. Figures show error bars or seed scatters where useful. Reported effects are either large relative to seed noise or visibly null: greedy surplus near one grid step, \unop{} surplus of order $\alpha$, and MovieLens collapse.

Compute.
All of the above is {NumPy}/{SciPy}/{scikit-learn} on {CPU}. We used $96$ {CPU} processes. The original controlled suite is $4508$ rows. Public complete files add $6601$ MovieLens-100K jobs, $42{,}280$ MovieLens-1M jobs, and $144$ runs ($16$ seeds $\times$ $9$ policies) each for Adult, German, {HELOC}, Credit Card, Give Me Some Credit, and Bank Marketing. All completed with zero failures.

\section{Bimatrix games}
\subsection{Bimatrix games, one at a time}
\label{app:bimatrix}

Figure~\ref{fig:A-bimatrix-heat} is the predictability side of the same data. Table~\ref{tab:A-bimatrix-games} reports mean human payoff in each of the ten games, averaged over the three predictors and twelve seeds. The ranking of policies is not a universal leaderboard. In matching pennies and rock, paper, scissors, greedy is a gift to the opponent and the mixed policies sit near or above zero. In coordination and stag hunt, greedy, small-$\alpha$ \unop{}, and \aunop{} are tied at the top, and uniform is strictly worse. In prisoner's dilemma every near-optimal policy collapses to defection, which is exactly what an adaptive counterpart does to a one-step cooperator.

\begin{table}[htbp]
\caption{Mean human payoff by bimatrix game, averaged over seeds and predictors. Mixing is mandatory in zero-sum, optional in general-sum, and costly in common-interest.}
\label{tab:A-bimatrix-games}
\centering
\begin{center}
\small
\setlength{\tabcolsep}{5.5pt}
\renewcommand{\arraystretch}{1.15}
\begin{tabular}{lccccc}
\toprule
Game & Greedy & $\varepsilon$-greedy & UNOP $\alpha{=}0.15$ & Aware-UNOP & Uniform \\
\midrule
Matching pennies & \cellcolor{cellblue}$-0.146$ & -0.125 & -0.000 & \cellcolor{cellgreen}$0.055$ & 0.011 \\
RPS & \cellcolor{cellblue}$-0.141$ & -0.113 & 0.000 & \cellcolor{cellgreen}$0.053$ & 0.012 \\
Prisoner's dilemma & 1.111 & 1.061 & 1.111 & 1.111 & 0.590 \\
Chicken & 1.765 & 1.616 & 1.818 & 1.871 & 1.382 \\
Inspection & 0.921 & 0.943 & 0.945 & 0.946 & 0.054 \\
Shapley & 0.339 & 0.345 & 0.391 & 0.412 & 0.333 \\
Coordination & \cellcolor{cellgreen}1.951 & 1.858 & \cellcolor{cellgreen}1.951 & \cellcolor{cellgreen}1.951 & 0.997 \\
Stag hunt & \cellcolor{cellgreen}4.878 & 4.789 & \cellcolor{cellgreen}4.878 & \cellcolor{cellgreen}4.878 & 3.959 \\
Battle of the sexes & 1.163 & 1.081 & 1.163 & 1.082 & 0.518 \\
Pure coordination & 0.965 & 0.902 & 0.964 & 0.964 & 0.329 \\
\bottomrule
\end{tabular}
\end{center}

\end{table}

\begin{figure}[t]
\centering
\includegraphics[width=0.98\linewidth]{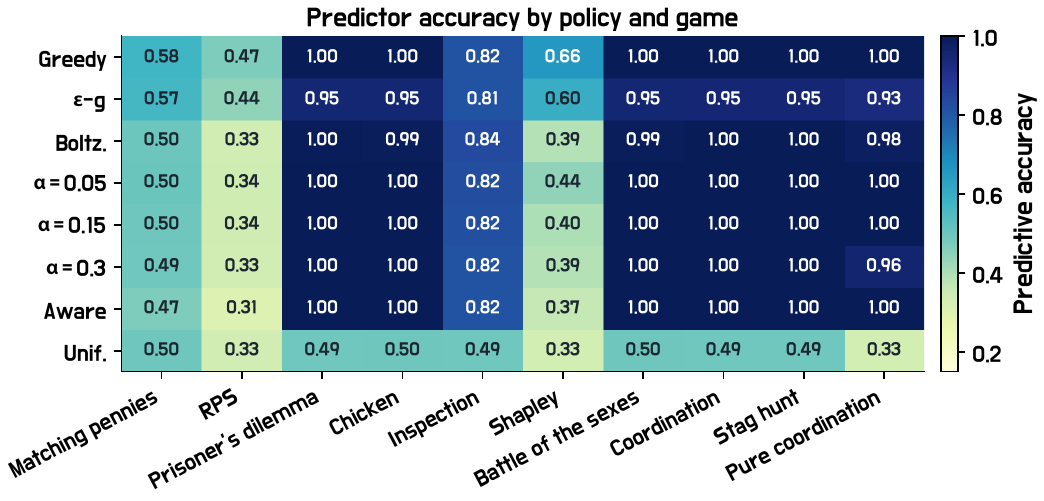}
\caption{Predictive accuracy by policy in the rows and game in the columns. Each cell is the numeric mean accuracy. Darker is more readable. Greedy is readable. Uniform is not. \unop{} lands in between only where the game actually has room at the chosen $\alpha$.}
\label{fig:A-bimatrix-heat}
\end{figure}

A reasonable worry is that the three-case pattern comes only from using a frequency predictor. Figure~\ref{fig:A-predictors} splits the same points by predictor. Frequency, Markov, and online softmax all produce the same qualitative cloud. Table~\ref{tab:A-bimatrix-T} reruns greedy, \unop{}-$\alpha=0.15$, and \aunop{} at $T\in\{200,500,1000,2000\}$ with a frequency predictor. In zero-sum games the mixed policies stay near the value as $T$ grows, while greedy does not recover. In common-interest games the three policies that can coordinate stay coordinated. The three-case pattern is not a short-horizon illusion.

\begin{figure}[t]
\centering
\includegraphics[width=0.98\linewidth]{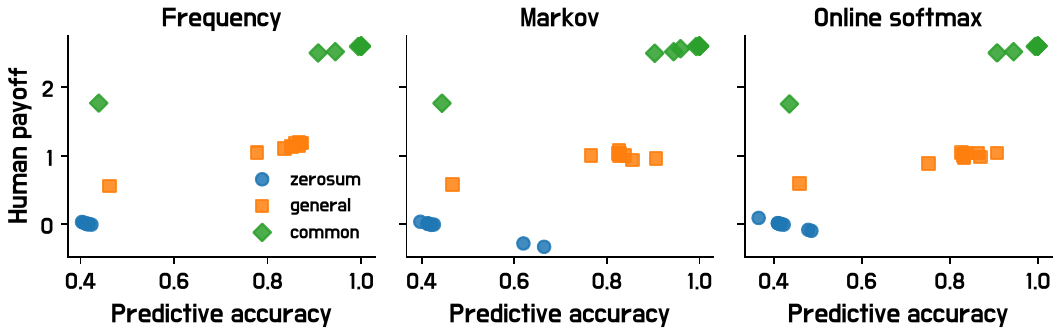}
\caption{Effectiveness against predictive accuracy, split by the optimizer's predictor. Circles are zero-sum, squares general-sum, diamonds common-interest. The three-case pattern is not an accident of counting frequencies.}
\label{fig:A-predictors}
\end{figure}

\begin{table}[htbp]
\caption{Human payoff versus horizon, grouped by game class under a frequency predictor. Zero-sum mixing remains necessary. Common-interest mixing remains a tax.}
\label{tab:A-bimatrix-T}
\centering
\begin{center}
\small
\setlength{\tabcolsep}{5.5pt}
\renewcommand{\arraystretch}{1.15}
\begin{tabular}{llcccc}
\toprule
Class & Policy & $T{=}200$ & $T{=}500$ & $T{=}1000$ & $T{=}2000$ \\
\midrule
Zero-sum & Greedy & -0.004 & -0.011 & -0.008 & -0.004 \\
 & UNOP $\alpha{=}0.15$ & 0.035 & -0.001 & -0.002 & 0.009 \\
 & Aware-UNOP & \cellcolor{cellgreen}0.021 & \cellcolor{cellgreen}0.017 & \cellcolor{cellgreen}0.033 & \cellcolor{cellgreen}0.013 \\
General-sum & Greedy & 1.206 & 1.196 & 1.197 & 1.201 \\
 & UNOP $\alpha{=}0.15$ & 1.232 & 1.224 & 1.226 & 1.225 \\
 & Aware-UNOP & 1.240 & 1.228 & 1.228 & 1.228 \\
Common-interest & Greedy & 2.625 & 2.607 & 2.596 & 2.595 \\
 & UNOP $\alpha{=}0.15$ & 2.626 & 2.607 & 2.596 & 2.595 \\
 & Aware-UNOP & 2.626 & 2.607 & 2.596 & 2.595 \\
\bottomrule
\end{tabular}
\end{center}

\end{table}

In matching pennies the human's actions are heads and tails, the optimizer wants to mismatch, and a greedy human who has a tiny preference for heads, or who simply broke a tie the same way twice, becomes a pure strategy that the optimizer can sit on. \unop{} with $\alpha$ large enough to include both actions is just the mixed security strategy, which is why its payoff sits at $0.000$. \aunop{} does a little better ($+0.054$) because it is not mixing independently of $P_t$: it puts extra mass on whichever action the predictor currently underweights, and is therefore anti-correlated with a slowly adapting opponent. The gain is empirical against that learner. It is not a guarantee against an exact best responder.

\section{Posted pricing}
\subsection{Posted pricing: where the surplus actually goes}
\label{app:price-long}

The main-text pricing figure reports surplus and the effectiveness-predictability plane. Figure~\ref{fig:A-pricing-scatter} shows the same runs as a scatter against the hidden value $v$. Greedy buyers lie on a tight band near zero surplus for every $v$: the monopolist found the threshold. \unop{} buyers keep a wedge of surplus. Uniform and Boltzmann buyers go negative: they are sometimes paying more than $v$. Figure~\ref{fig:ts} is the time series: greedy posted prices climb toward $v$. \unop{} stays below $v$. The hit rate should be compared with Bayes accuracy, not read as a failure to learn the mixture. Seed-level spreads appear as error bars in Figure~\ref{fig:ts}. Table~\ref{tab:main} reports means only, and the ordering is not an average of a bimodal mess.

\begin{figure}[t]
\centering
\includegraphics[width=0.98\linewidth]{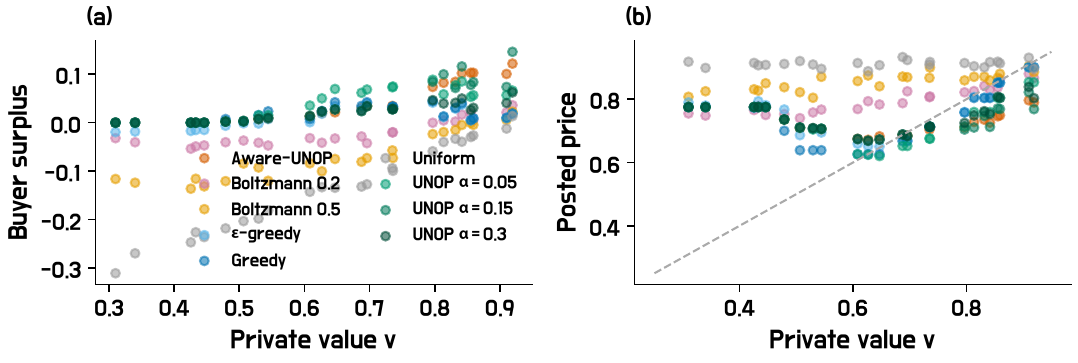}
\caption{Left: buyer surplus against private value. Right: posted price against $v$, with the diagonal $p=v$ as a dashed reference. Greedy transactions hug the diagonal from below. \unop{} leaves a visible gap. Boltzmann and uniform cross the diagonal.}
\label{fig:A-pricing-scatter}
\end{figure}

\begin{figure}[t]
\centering
\includegraphics[width=0.98\linewidth]{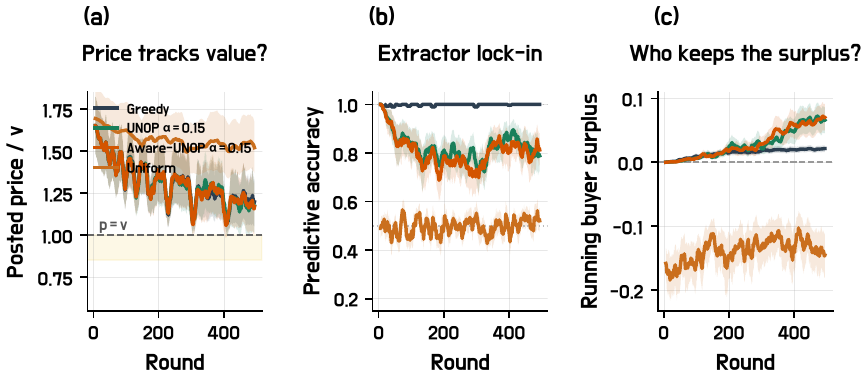}
\caption{Pricing dynamics (rolling mean $\pm$ s.e.). (a)~Posted price relative to $v$. (b)~Predictor accuracy. (c)~Running buyer surplus. Greedy prices climb toward $v$. \unop{} stays lower. Accuracy for \unop{} is bounded by the policy's Bayes accuracy.}
\label{fig:ts}
\end{figure}

\subsection{A worked numerical example}
\label{app:price-worked}

Suppose $v=0.70$ and the price grid is $\{0.00,0.05,\ldots,1.00\}$, so $\Delta=0.05$. A greedy buyer purchases at every $p\le 0.70$ and refuses otherwise. After enough plays, {UCB} on revenue will prefer $p=0.70$ to $p=0.65$ once both empirical buy rates are near one, because $0.70\cdot 1>0.65\cdot 1$, and it will prefer $p=0.70$ to $p=0.75$ once the latter's empirical buy rate is near zero. Buyer surplus is then $0$ when the posted price is exactly $0.70$, or $0.05$ if the optimizer sits one cell lower. In the main runs the mean leftover is $0.018$, which is a fraction of a cell, as Proposition~\ref{prop:price} said.

Now take \unop{} with $\alpha=0.15$ and reservation $0$. At $p=0.50$, buy gives $0.20$ and skip gives $0$, so $u^\star=0.20$ and $u^\star-\alpha=0.05$. Skip lies below that cut and buy is the unique $\alpha$-optimal {IR} action, hence $\PP(\mathrm{buy})=1$. At $p=0.60$, buy gives $0.10$ and skip gives $0$. Now $u^\star-\alpha=-0.05$, so both actions lie in $\cA_{0.15}$ and both are {IR}, and the human mixes equally. At $p=0.65$, buy gives $0.05$ and skip gives $0$. Both remain inside $\cA_{0.15}$ and both are {IR}, so the mix is again equal. At $p=0.75$, buy gives $-0.05$, which is excluded by the reservation even though a Boltzmann policy would still buy with positive probability. A seller who wants demand near $1$ is pushed down to $p\le 0.55$, and the buyer keeps about $0.15$.

\subsection{Grid resolution and horizon}
\label{app:grid}

Tables~\ref{tab:A-grid} and~\ref{tab:A-horizon} vary grid resolution and horizon. As the price grid goes from $11$ to $41$ cells at fixed $T=500$, greedy surplus falls ($0.031$ at $11$ cells, $0.006$ at $41$), which is $\Theta(\Delta)$. \unop{}-$\alpha=0.15$ remains above greedy at every resolution. Boltzmann and uniform stay negative. As $T$ grows on the default $21$-cell grid, greedy surplus is already small at $T=150$ and does not grow. \unop{} surplus rises ($0.007$ at $150$, $0.061$ at $800$), because post-burn-in means for \unop{} improve once early {UCB} exploration is behind it. \unop{} remains above greedy at every horizon in the sweep.

The two sweeps should be read together. A finer grid is a larger action space for the monopolist: each additional cell needs samples before its revenue estimate is usable, so at fixed horizon a finer spacing leaves more leftover exploration and compresses surplus after the early rounds for every policy. Extending $T$ gives the same optimizer time to finish that exploration, which is why surplus recovers as $T$ grows even though the spacing is held fixed. In short, finite-horizon surplus tracks how far the extractor has gotten through its action set, not only how wide the human's $\alpha$-band is. The practical boundary is that \unop{} still beats greedy on every grid we tried, but a very fine menu paired with a short horizon will look closer to the greedy grid step than the large-sample $\Theta(\alpha)$ band.

\begin{table}[htbp]
\caption{Buyer surplus versus the number of cells on the price grid.}
\label{tab:A-grid}
\centering
\begin{center}
\small
\setlength{\tabcolsep}{5.5pt}
\renewcommand{\arraystretch}{1.15}
\begin{tabular}{lcccc}
\toprule
Policy & $11$ cells & $21$ cells & $31$ cells & $41$ cells \\
\midrule
Greedy & 0.031 & 0.013 & 0.010 & 0.006 \\
UNOP $\alpha{=}0.15$ & \cellcolor{cellgreen}0.065 & \cellcolor{cellgreen}0.036 & \cellcolor{cellgreen}0.021 & \cellcolor{cellgreen}0.012 \\
Boltzmann 0.2 & -0.010 & -0.022 & -0.028 & -0.030 \\
Uniform & -0.140 & -0.135 & -0.143 & -0.142 \\
\bottomrule
\end{tabular}
\end{center}

\end{table}

\begin{table}[htbp]
\caption{Buyer surplus versus horizon $T$. Greedy tracks the grid step. \unop{} remains above it.}
\label{tab:A-horizon}
\centering
\begin{center}
\small
\setlength{\tabcolsep}{5.5pt}
\renewcommand{\arraystretch}{1.15}
\begin{tabular}{lcccc}
\toprule
Policy & $T{=}150$ & $T{=}300$ & $T{=}500$ & $T{=}800$ \\
\midrule
Greedy & 0.006 & 0.011 & 0.014 & 0.016 \\
UNOP $\alpha{=}0.15$ & 0.007 & 0.020 & \cellcolor{cellgreen}0.042 & \cellcolor{cellgreen}0.061 \\
Boltzmann 0.2 & -0.039 & -0.030 & -0.020 & -0.011 \\
Uniform & -0.150 & -0.139 & -0.135 & -0.126 \\
\bottomrule
\end{tabular}
\end{center}

\end{table}

\subsection{The individual-rationality ablation}
\label{app:ir}

\unop{} as defined in Section~\ref{sec:method} intersects $\cA_\alpha$ with $\{a:u(a)\ge \underline{u}\}$ whenever that intersection is nonempty. In pricing $\underline{u}=0$ (skip). In recommendation $\underline{u}=\tau$. Table~\ref{tab:A-ir} turns that intersection off and reruns the same policies.

Without the floor, \unop{}-$\alpha=0.15$ in pricing goes from $+0.036$ to $-0.018$, and \aunop{} goes from $+0.039$ to $-0.017$. This is a separate smaller-seed ablation, while Table~\ref{tab:main} reports $0.040$ with the floor on. That sign change is the whole argument for having a floor. Mixing into buy-above-value raises one-step entropy and also hands revenue to a learning monopolist. Greedy, $\varepsilon$-greedy, Boltzmann, and uniform are unaffected by the switch, because they never used the floor. In recommendation the same switch costs about $0.04$ of surplus for \unop{}.

\begin{table}[htbp]
\caption{Effectiveness with and without the individual-rationality floor. The floor is unused for greedy and for policies that already mix everywhere. For \unop{} it separates preserved surplus from a teaching signal.}
\label{tab:A-ir}
\centering
\begin{center}
\small
\setlength{\tabcolsep}{5.5pt}
\renewcommand{\arraystretch}{1.15}
\begin{tabular}{lcccc}
\toprule
Policy & \multicolumn{2}{c}{Pricing surplus} & \multicolumn{2}{c}{Recsys surplus} \\
 & IR on & IR off & IR on & IR off \\
\midrule
Greedy & 0.012 & 0.012 & 0.050 & 0.050 \\
UNOP $\alpha{=}0.15$ & \cellcolor{cellgreen}0.036 & -0.018 & \cellcolor{cellgreen}0.136 & 0.094 \\
Aware-UNOP & \cellcolor{cellgreen}0.039 & -0.017 & 0.105 & 0.055 \\
Boltzmann 0.2 & -0.022 & -0.022 & 0.037 & 0.037 \\
Uniform & -0.135 & -0.135 & -0.079 & -0.079 \\
\bottomrule
\end{tabular}
\end{center}

\end{table}

The floor is a modeling claim about extractors, not a claim that people compute a reservation and then mix above it. If a policy mixes into actions it would never take against a fixed menu, those actions teach the extractor where the cutoff still is. The reservation floor is what blocks that lesson.

\subsection{IR-constrained baselines}
\label{app:ir-base}

A natural question is whether the gap between \unop{} and Boltzmann/$\varepsilon$-greedy is entirely the reservation floor. Table~\ref{tab:A-ir-base} applies the same floor to those baselines: $\varepsilon$-greedy and Boltzmann are renormalized on $\{a:u(a)\ge\underline{u}\}$ whenever that set is nonempty. In pricing the floor turns Boltzmann's negative surplus positive, but IR-Boltzmann $0.2$ still trails \unop{}-$\alpha=0.15$ ($0.028$ vs.\ $0.040$). IR-$\varepsilon$-greedy remains near greedy. In recommendation, IR-Boltzmann $0.2$ reaches $0.113$ surplus, still below \unop{} at $0.148$. The floor removes teaching-signal losses. Uniform mixing on the near-optimal set is what keeps the extractor off the threshold.

\begin{table}[htbp]
\caption{Unconstrained vs.\ IR-constrained baselines on the same pricing and recsys loops as the main text ($T=500$/$300$).}
\label{tab:A-ir-base}
\centering
\begin{tabular}{lcccc}
\hline
Policy & \multicolumn{2}{c}{Pricing} & \multicolumn{2}{c}{Recsys} \\
 & surplus & acc. & surplus & acc. \\
\hline
Greedy & $0.018$ & $0.999$ & $0.052$ & $0.704$ \\
$\varepsilon$-greedy & $0.011$ & $0.954$ & $0.051$ & $0.489$ \\
IR-$\varepsilon$-greedy & $0.017$ & $0.970$ & $0.058$ & $0.564$ \\
Boltzmann $0.2$ & $-0.021$ & $0.665$ & $0.024$ & $0.169$ \\
IR-Boltzmann $0.2$ & $0.028$ & $0.766$ & $0.113$ & $0.389$ \\
Boltzmann $0.5$ & $-0.068$ & $0.595$ & $-0.043$ & $0.134$ \\
IR-Boltzmann $0.5$ & $0.025$ & $0.738$ & $0.090$ & $0.440$ \\
UNOP $\alpha{=}0.15$ & \cellcolor{cellgreen}$0.040$ & $0.807$ & \cellcolor{cellgreen}$0.148$ & $0.340$ \\
Uniform & $-0.133$ & $0.499$ & $-0.094$ & $0.107$ \\
\hline
\end{tabular}

\end{table}

\subsection{Misspecified utilities, and how much $\alpha$ has to absorb}
\label{app:noise}

The theory treats $u$ as known. In any honest deployment the human has $\hat u=u+\xi$. Table~\ref{tab:A-noise} adds independent Gaussian noise of scale $\sigma\in\{0,0.02,0.05,0.10,0.20\}$ to the utility vector before the policy is applied. In pricing, \unop{}-$\alpha=0.15$ surplus falls from $0.032$ at $\sigma=0$ to $-0.018$ at $\sigma=0.20$. The last of those is noise larger than $\alpha$, so the $\alpha$-set constructed from $\hat u$ is no longer a slightly wider true $\alpha$-set. Raising $\alpha$ further at $\sigma=0.20$ does not recover surplus in our runs: at the same noise, $\alpha\in\{0.25,0.30,0.40,0.50\}$ stays negative and slowly worsens, because a wider noisy band absorbs more wrongly signed actions through the reservation check. When $\sigma$ is small relative to $\alpha$, a default $\alpha$ on the order of twice the plausible utility noise is safer than a default that tries to be as small as the worst-case bound would suggest. When $\sigma$ exceeds that scale, the binding fix is a better $\hat u$, not a larger band.

\begin{table}[htbp]
\caption{Effectiveness as a function of independent Gaussian noise on the human's utility vector. \unop{} remains useful while $\sigma$ is small relative to $\alpha$. At $\sigma=0.20$ a larger $\alpha$ does not restore pricing surplus.}
\label{tab:A-noise}
\centering
\begin{center}
\small
\setlength{\tabcolsep}{5.5pt}
\renewcommand{\arraystretch}{1.15}
\begin{tabular}{llccccc}
\toprule
Env. & Policy & $\sigma{=}0.00$ & $\sigma{=}0.02$ & $\sigma{=}0.05$ & $\sigma{=}0.10$ & $\sigma{=}0.20$ \\
\midrule
Pricing & Greedy & 0.013 & 0.020 & 0.018 & 0.011 & -0.012 \\
 & UNOP $\alpha{=}0.15$ & \cellcolor{cellgreen}0.032 & 0.013 & 0.010 & 0.001 & -0.018 \\
 & Uniform & -0.135 & -0.136 & -0.136 & -0.136 & -0.136 \\
Recsys & Greedy & 0.044 & 0.090 & 0.109 & 0.106 & 0.065 \\
 & UNOP $\alpha{=}0.15$ & \cellcolor{cellgreen}0.129 & 0.106 & 0.101 & 0.094 & 0.054 \\
 & Uniform & -0.076 & -0.081 & -0.081 & -0.081 & -0.081 \\
\bottomrule
\end{tabular}
\end{center}

\end{table}

Recommendation is gentler, because a typical recommendation list has several items whose qualities are already close. Even there, $\sigma=0.20$ is too much for $\alpha=0.15$. We do not claim a sharp phase transition. We claim that a tiny $\alpha$ is not enough once noise exceeds slack, and that widening the band alone is not a recovery recipe.

\section{Recommendation and public data}
\subsection{Margin-seeking recommendation lists, catalog size, and promoted exposure}
\label{app:rec-long}

Figure~\ref{fig:A-recsys-scatter} plots promoted-item click rate against user surplus for every synthetic recsys seed. Greedy seeds cluster at high promo and low surplus, \unop{} seeds sit at lower promo and higher surplus, and uniform seeds occupy the high-promo, negative-surplus corner. Promoted exposure falls from $0.858$ to $0.330$ at $\alpha=0.15$ not because we penalized ads, but because the items that survive as unique-best offers are less often the promoted ones when $\tau$ is a band.

\begin{figure}[t]
\centering
\includegraphics[width=0.78\linewidth]{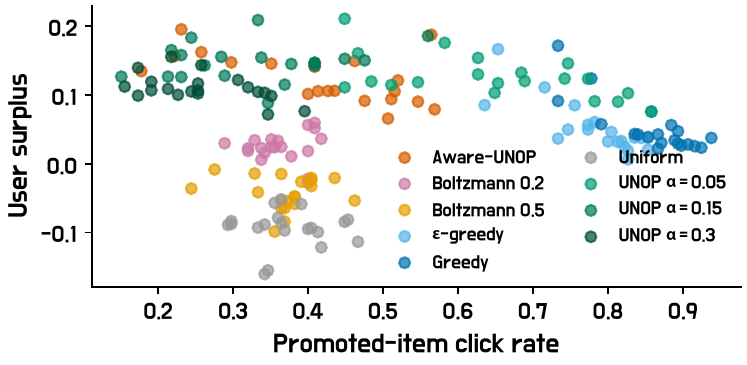}
\caption{Each point is one synthetic recommender seed. Horizontal: fraction of clicks on promoted items. Vertical: user surplus $q-\tau$.}
\label{fig:A-recsys-scatter}
\end{figure}

Table~\ref{tab:A-reclist} asks whether that cloud depends on having chosen $48$ items and a recommendation list of $8$. We crossed catalog sizes $\{24,48,72\}$ with recommendation-list sizes $\{4,8,12\}$ for greedy, \unop{}-$\alpha=0.15$, and uniform. Greedy stays poor across the grid. Uniform stays worse. \unop{} stays better, and it likes larger recommendation lists, because a larger recommendation list is more likely to contain a second item inside $\cA_\alpha$. A recommendation list of four with a unique-best offer and three filler items leaves little room: there may be nothing near-optimal to mix with except skip.

\begin{table}[htbp]
\caption{User surplus versus catalog size and recommendation-list size. The ranking does not depend on the particular $(48,8)$ pair used in the main text.}
\label{tab:A-reclist}
\centering
\begin{center}
\small
\setlength{\tabcolsep}{5.5pt}
\renewcommand{\arraystretch}{1.15}
\begin{tabular}{llccc}
\toprule
Policy & Slate & cat.\ $24$ & cat.\ $48$ & cat.\ $72$ \\
\midrule
Greedy & $4$ & 0.039 & 0.037 & 0.039 \\
 & $8$ & 0.055 & 0.055 & 0.064 \\
 & $12$ & 0.111 & 0.070 & 0.087 \\
UNOP $\alpha{=}0.15$ & $4$ & 0.125 & 0.092 & 0.052 \\
 & $8$ & 0.161 & \cellcolor{cellgreen}0.148 & 0.111 \\
 & $12$ & 0.302 & 0.198 & 0.152 \\
Uniform & $4$ & -0.091 & -0.131 & -0.154 \\
 & $8$ & -0.027 & -0.091 & -0.127 \\
 & $12$ & 0.024 & -0.049 & -0.093 \\
\bottomrule
\end{tabular}
\end{center}

\end{table}

\subsection{MovieLens, user by user, on the complete files}
\label{app:ml-long}

The main text reports that on MovieLens, \unop{} collapses to greedy because five-star ratings do not have slack. Figure~\ref{fig:A-ml100k} is the user-level answer on the official complete 100K file: the scatter hugs the diagonal, and the per-user difference is a spike at zero. Table~\ref{tab:A-ml-extra} records the same collapse on MovieLens-1M ($6{,}040$ users) and the promoted-item rates on both files. Greedy and \unop{} remain close. Uniform is much higher. This is the same collapse on the complete file, without subsampling.

\begin{figure}[t]
\centering
\includegraphics[width=0.92\linewidth]{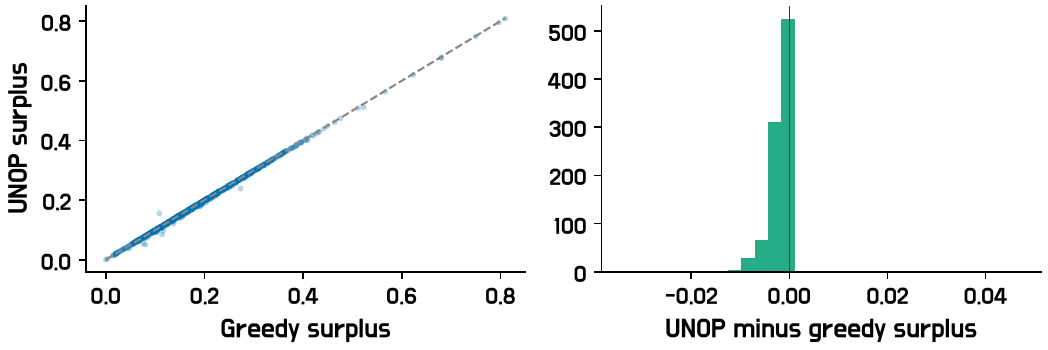}
\caption{Complete MovieLens-100K, all $943$ users. Left: per-user greedy versus \unop{}. Right: per-user difference. The collapse is not an average over a hidden subpopulation.}
\label{fig:A-ml100k}
\end{figure}

\begin{table}[htbp]
\caption{Complete MovieLens files: surplus collapse and promoted exposure. UNOP uses $\alpha=0.1$.}
\label{tab:A-ml-extra}
\centering
\begin{center}
\small
\setlength{\tabcolsep}{5.5pt}
\renewcommand{\arraystretch}{1.15}
\begin{tabular}{lccccc}
\toprule
File & $N$ & greedy surplus & UNOP surplus & greedy promo & UNOP promo \\
\midrule
MovieLens-100K & 943 & 0.189 & 0.187 & 0.085 & 0.096 \\
MovieLens-1M & 6040 & 0.187 & 0.185 & 0.065 & 0.084 \\
\bottomrule
\end{tabular}
\end{center}

\end{table}

Why is there no slack? Because a user who has rated twenty movies on a five-star grid often has a unique five-star title, and after scaling to $[0,1]$ that title is $1.0$ and the four-star titles are $0.75$, so any $\alpha<0.25$ leaves a singleton $\cA_\alpha$. If a platform displayed utilities on a finer scale, such as a continuous rating, a dwell-time proxy, or a willingness-to-watch slider, the same protocol would look like the synthetic recsys panel. The public ratings are coarse, and the reported result is the one those ratings support.

\subsection{UCI Adult and German Credit}
\label{app:adult}

Strategic classification on real tables is the other public-data pillar. The institution starts from a logistic classifier trained on unmanipulated official records (Adult accuracy $0.852$ on all $48{,}842$ rows after one-hot to $d=108$, and German Credit $0.785$ on all $1{,}000$ rows, $d=61$), then continues with online logistic updates on the features that agents actually present. Agents may game a small set of numeric coordinates that are at least arguably effort-like: age, education-num, capital gain/loss, and hours on Adult, and duration, credit amount, age, and existing credits on German. The action is one of five intensities along $\mathrm{sign}(w)$ on those coordinates, at quadratic cost. The predictor is a frequency model over those five intensities. This is not a claim that people in the 1994 Adult file actually gamed their hours-per-week against an adaptive classifier. It is a claim that if they did, with a linear institution that maximizes accuracy, mixing would reduce readability without raising utility.

Table~\ref{tab:public} already reports agent utility and predictor accuracy on the complete {UCI} Adult and German Credit tables. Greedy is best or tied for best on utility. On Adult, \unop{}-$\alpha=0.05$ gives utility/accuracy $0.181$/$0.428$, versus $0.178$/$0.482$ for greedy. Larger $\alpha$ costs utility without an extraction-flavored reward. Mixing does reduce readability. It just does not buy the agent anything the accuracy-seeking institution was threatening to take.

Figure~\ref{fig:A-cost-lockk} (left) varies the quadratic cost coefficient on the synthetic strategic population. As cost rises, everyone games less and policy differences shrink. No cost in this family makes \unop{} outperform greedy by a margin comparable to pricing. The counterpart maximizes accuracy rather than surplus.

\begin{figure}[t]
\centering
\includegraphics[width=0.98\linewidth]{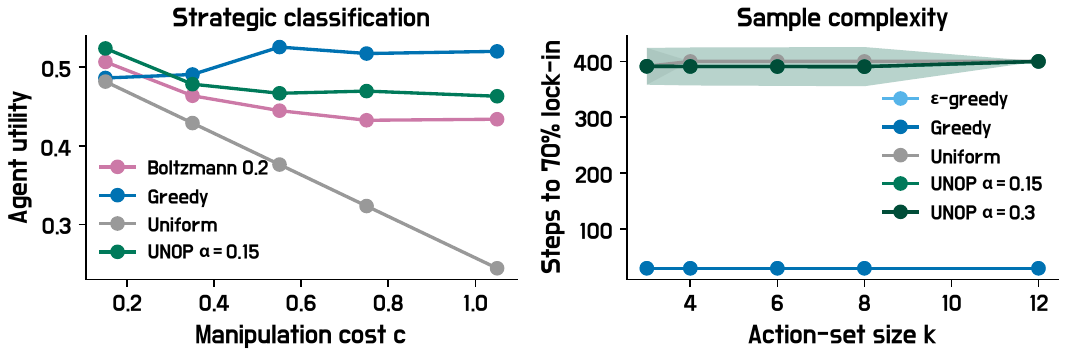}
\caption{Left: synthetic strategic classification as the manipulation-cost coefficient $c$ varies. Right: steps to $70\%$ lock-in as the action-set size $k$ varies. Greedy and $\varepsilon$-greedy lock in $30$ steps at every $k$ (the window length). \unop{} and uniform do not lock within the horizon.}
\label{fig:A-cost-lockk}
\end{figure}

\section{Lock-in and other extractors}
\subsection{Lock-in, populations, and the public-good reading}
\label{app:lock-long}

Theorem~\ref{thm:lockin} is a Dirichlet multinomial calculation, and Figure~\ref{fig:A-cost-lockk} (right) is the empirical version as $k=|A|$ grows. This separate $k$-sweep uses $T=400$. The main lock-in experiment uses $T=350$. Greedy and $\varepsilon$-greedy hit $70\%$ rolling accuracy at $t=30$ for every $k$ we tried. Uniform almost never hits it within the sweep horizon. \unop{} at $\alpha=0.15$ tracks uniform on that random utility draw once $\alpha$ admits several near-ties. The lock-in figure (Figure~\ref{fig:lock}) used a vector whose near-ties were farther apart, so small $\alpha$ still leaves a peak.

Figure~\ref{fig:lock} also shows the population overlay. A shared margin-seeking optimizer trained on a mix of greedy and \unop{} users has an uneven accuracy curve: highest on a pure-greedy population and lowest around a $30\%$ \unop{} minority. Greedy users' welfare rises when they are mixed with \unop{} neighbors.

\begin{figure}[t]
\centering
\includegraphics[width=0.98\linewidth]{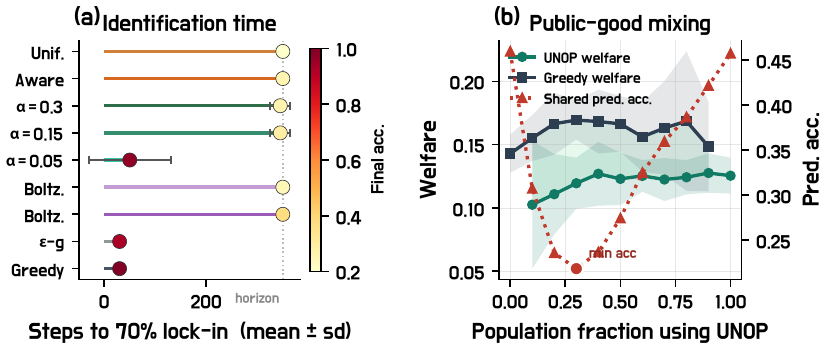}
\caption{Left: steps to $70\%$ lock-in (mean $\pm$ sd) with marker color $=$ final accuracy. Dotted line is the horizon. Right: welfare of greedy and \unop{} users in a shared-optimizer population (shaded gap) and the extractor's accuracy (right axis, min marked). Heterogeneity is itself a defense.}
\label{fig:lock}
\end{figure}

\subsection{Two tradeoffs, Gibbs leak, and why temperature is the wrong choice}
\label{app:frontiers-long}

Figure~\ref{fig:pareto} already contrasts the two tradeoffs on a six-action utility vector. The same qualitative gap (smooth Gibbs leak versus UNOP step support) appears on random vectors of dimension $4$ through $16$ in the dump \texttt{pareto.json}. Figure~\ref{fig:A-support-leak} makes that picture concrete. The left panel is the mean size of $\cA_\alpha$ as a function of $\alpha$. The right panel is the amount of Gibbs mass that sits outside $\cA_\alpha$ as a function of $\alpha$, for several temperatures. At $\tau=0.2$ and $\alpha=0.15$, more than a quarter of the Gibbs mass is already off the near-optimal set. Temperature is not a substitute for a hard support constraint. The two-panel $\alpha$-sweep in Figure~\ref{fig:alpha} is the corresponding empirical picture: effectiveness rises and then falls, and accuracy falls as soon as $\cA_\alpha$ is non-singleton.

\begin{figure}[t]
\centering
\includegraphics[width=0.98\linewidth]{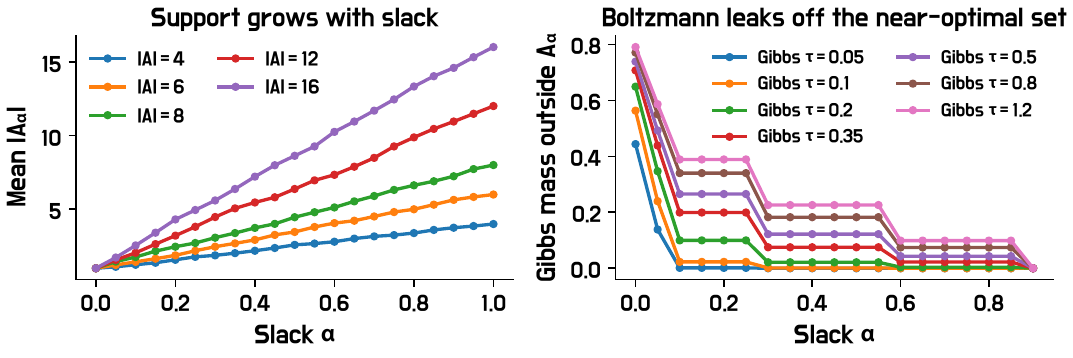}
\caption{Left: mean size of the $\alpha$-optimal set on random utility vectors. Right: Gibbs mass lying outside $\cA_\alpha$. Temperature is not a substitute for a hard support constraint.}
\label{fig:A-support-leak}
\end{figure}

\subsection{A different extractor: Thompson sampling}
\label{app:optvar}

{UCB} is a particular way to turn uncertain demand into a posted price. Table~\ref{tab:A-optimizer} reruns the pricing environment with a Beta-Bernoulli Thompson monopolist that posts $\arg\max_p p\cdot \theta_p$ for $\theta_p\sim\mathrm{Beta}(a_p,b_p)$, and includes Hedge as an additional human baseline. The qualitative ranking does not flip: uniform is still a disaster, greedy is still extracted, \unop{} still keeps surplus. The \unop{} surplus is larger against Thompson ($0.139$) than against {UCB} ($0.036$) in this parameterization, because Thompson is more willing to post inside a noisy band. Hedge is near greedy against {UCB} and slightly negative against Thompson, consistent with mean-based learning not being a defense against a posted-price extractor.

\begin{table}[htbp]
\caption{Buyer surplus against a {UCB} monopolist and a Thompson monopolist, same price grid and values of $v$. The human-side ranking is stable. Hedge is not a substitute for \unop{}.}
\label{tab:A-optimizer}
\centering
\begin{center}
\small
\setlength{\tabcolsep}{5.5pt}
\renewcommand{\arraystretch}{1.15}
\begin{tabular}{lcc}
\toprule
Policy & UCB surplus & Thompson surplus \\
\midrule
Greedy & 0.012 & 0.022 \\
UNOP $\alpha{=}0.15$ & \cellcolor{cellgreen}0.036 & \cellcolor{cellgreen}0.139 \\
Uniform & -0.135 & -0.132 \\
Hedge & 0.003 & -0.003 \\
\bottomrule
\end{tabular}
\end{center}

\end{table}

\subsection{Putting the environments together}
\label{app:tradeoff}

The environments do not share a utility scale, so a single scatter is unreadable. The main-text three-case figure, the pricing/recsys tables, and Table~\ref{tab:public} already locate each setting. If \unop{} were simply the random policy, it would sit with uniform. If it were simply greedy with a different name, it would sit with greedy. On pricing and recsys it sits up and left of greedy. On both MovieLens files it sits on top of greedy. On Adult and German it sits down and left: less readable, not more effective. Predictability is correlated with policy entropy across environments, but two policies with the same entropy can be very different for the human if one spends that entropy inside $\cA_\alpha$ and the other spends it on dominated actions. Accuracy, optimizer payoff, and human payoff are therefore the reported outcomes. Entropy alone does not determine them.

\subsection{A second worked example: a margin-seeking recommendation list}
\label{app:reclist-worked}

Imagine four items with qualities $(0.90, 0.72, 0.55, 0.40)$ and a reservation $\tau=0.60$. The platform can show three of them. A unique-best margin-seeking recommendation list that has identified $\tau$ will offer the item with quality $0.72$ as the unique best (it is just above $\tau$) and fill with $0.55$ and $0.40$ as filler items, because $0.90$ would give the user a large surplus and a lower margin. A greedy user clicks $0.72$ every time, surplus $0.12$, and the platform has sat on $\tau$. An \unop{} user with $\alpha=0.15$ sees utilities $(0.12, -0.05, -0.20, 0)$ for the three items and skip. The $\alpha$-optimal {IR} set is $\{0.72\text{-item}, \text{skip}\}$, because skip is $0$ and the offer is $0.12$, which differ by less than $0.15$, while the filler items are below reservation. The user therefore mixes equally between taking the offer and skipping. The platform now sees a $50\%$ acceptance rate on an item that would have been a sure thing if $\tau$ were $0.60$ with a greedy user, and {UCB} cannot treat $0.72$ as an identified threshold. If the catalog had also included an organic $0.85$ item on the recommendation list, that item would have entered $\cA_\alpha$ as well, and the mix would have included a high-quality click. Withholding is how the platform tries to prevent that. Mixing on skip is how the user refuses to certify the withheld menu.

Uniform, on the same recommendation list, clicks the $0.40$ item with positive probability, which is worse than skip, and the platform learns something even less helpful to the user: that junk is sometimes acceptable. Unpredictability without a floor is not a defense.

\section{Limitations and further checks}
\subsection{Limitations}
\label{app:not-claim}

The response-curve commitment effect applies when the platform optimizes against that curve, when feasible menus leave more than one acceptable action, and when the utility estimate places the set correctly. Noise large enough to move the estimated set removes the surplus, which is why the pricing gain and the recommendation gain should be read together with the noise and singleton checks in Section~\ref{sec:results}.

We list explicitly what this extra empirical work does not establish. We have not shown that \unop{} cannot be exploited by a patient optimizer who knows the policy class and chooses a best response to the mixture in closed form. We have shown that the adaptive {UCB} and margin-seeking recommendation-list extractors used in the experiments cannot sit on a band the way they sit on a threshold. We have not shown that users will adopt $\alpha$ as a slider. We have shown that if they did, the effect of changing $\alpha$ rises and then falls, rather than more randomness always being better. We have not shown that MovieLens would look like the synthetic recsys panel under a continuous rating. We have shown that under the ratings the file actually contains, there is no slack to mix. We have not shown that Adult applicants should randomize their hours. We have shown that against an accuracy-maximizing logistic institution, randomization buys a policy that is harder to read, and not more utility. We have not run a user study or queried a deployed API. Every optimizer and every human policy in this paper is a simulator.

Across these checks, mixing helps inside slack and with a reservation. It does not help when the counterpart is a collaborator or a pure accuracy maximizer. A five-star grid does not create a non-singleton $\alpha$-band.

\subsection{The ten bimatrix games}
\label{app:ten-games}

This subsection records each bimatrix game in words, together with the mean payoff over predictors and seeds. The payoffs are those in \texttt{src/envs.py}. The numbers are means from \texttt{bimatrix.csv}.

Matching pennies.
The human wants to match, the optimizer wants to mismatch, and there is no slack in the interesting sense because against a best responder the two actions are equally good at mixed equilibrium and equally terrible if played pure. A greedy human who breaks a tie toward heads twice in a row has published a peak. Frequency counts move, the optimizer plays tails, and the human's continuation value is $-1$ until they mix again. In the runs, greedy finishes at about $-0.15$ rather than $-1$ because of the $5\%$ random slip and because identification is not instantaneous. The sign of the payoff is the result that matters. \unop{} with $\alpha$ large enough to include both actions is the mixed security strategy and finishes at $0.000$. \aunop{} finishes at $+0.054$, which is the gain from avoiding a slow predictor. Uniform is statistically the same as \unop{} here, because the $\alpha$-set is the whole action set. In this game, remaining effective requires remaining unpredictable.

Rock, paper, scissors.
The same conclusion holds with three actions. Greedy that locks onto rock is giving the optimizer paper. \unop{} that mixes on all three when they are near-ties is the usual security mix. The numbers track matching pennies closely (greedy about $-0.14$, \unop{} $0.00$, \aunop{} $+0.05$).

Prisoner's dilemma.
Both players have a dominant strategy to defect if the other is held fixed, and an adaptive counterpart that can predict cooperation will defect against it. Every near-optimal policy in our suite therefore collapses to defection and gets about $1.11$, the mutual-defection payoff, well below the cooperative payoff of $3$. \unop{} is not a cooperation device. If the human wants to cooperate against an extractor, they need a different theory (commitment, repeated-game arguments, outside enforcement), not a mixer on $\cA_\alpha$.

Chicken.
Swerving versus staying has an interior mixed equilibrium and two asymmetric pure equilibria. A greedy human who looks like they will stay can induce the optimizer to swerve, which is good for the human, until the optimizer's predictor becomes confident enough to stay instead. \unop{}-$\alpha=0.15$ edges greedy ($1.82$ versus $1.77$) because mixing in the near-tie region makes it harder for the optimizer to pick the punishing pure strategy and sit on it. The gap is small. In this general-sum game a little slack is mildly useful, neither mandatory nor forbidden.

Inspection.
A worker who always works is wasting effort when the inspector is not inspecting. A worker who always shirks is caught once the inspector learns that. The $\alpha$-set often contains both work and shirk against the empirical inspect rate, which is the support on which \unop{} mixes. \unop{}-$\alpha=0.05$ gets $0.96$ against greedy's $0.92$. The gap is small and in the direction predicted by the theory.

Shapley's game.
A cyclic general-sum game that is a standard test for learning dynamics because no pure equilibrium is stable. Mixing is almost unavoidable, and the policies that already mix do better than greedy, which keeps trying to sit on a disappearing best response.

Battle of the sexes.
Two pure equilibria, disagreement about which one, and a mixed equilibrium that both dislike. A greedy human who looks determined can sometimes pull the optimizer onto their preferred equilibrium, which is a case where being predictable is helpful even though the game is not common-interest. \unop{} does not destroy that, because small $\alpha$ still leaves a peak on the currently better equilibrium, but large $\alpha$ mixes toward the mixed equilibrium and pays the coordination tax of Proposition~\ref{prop:common}.

Coordination and pure coordination.
Greedy, small-$\alpha$ \unop{}, and \aunop{} all sit at the Pareto-optimal pure profile (payoff $1.95$ on the $2,0$ and $0,1$ game, $0.965$ on three-action pure coordination, $4.88$ on stag hunt, and suite mean $2.598$). Uniform falls to $1.00$ on the two-action game. Predictor accuracy is $1.0$ for the peaked policies. If a deployment cannot tell whether the counterpart is a collaborator or an extractor, \unop{} run by default against a helper misses the meeting point.

Stag hunt.
The Pareto-optimal profile is (stag, stag) at $5$, the safe profile is (hare, hare) at $3$, and a greedy human who is predicted to hunt stag will be met there by an adaptive partner. Mixing toward hare fails to reach the stag equilibrium. Our runs put greedy and \unop{} together at $4.88$ and uniform at $3.96$, which is the coordination tax. \unop{} is not an appropriate default when the counterpart is a cooperative assistant completing a joint task, which is the content of Proposition~\ref{prop:common}.

\subsection{Hyperparameters}
\label{app:hyper}

Table~\ref{tab:A-hyper} lists $\alpha$ and the other experimental settings.

\begin{table}[htbp]
\centering
\caption{Experimental settings. Only $\alpha$ is a user-facing parameter. The rest are simulation choices checked by ablations.}
\label{tab:A-hyper}
\small
\setlength{\tabcolsep}{5.5pt}
\renewcommand{\arraystretch}{1.12}
\begin{tabular}{p{0.20\linewidth}p{0.28\linewidth}p{0.44\linewidth}}
\toprule
Setting & Value(s) used & Role \\
\midrule
Bimatrix $T$ & $1000$ (ablation: $200$ to $2000$) & Long enough for frequency counts to concentrate. Ablations check that the ranking does not flip. \\
Bimatrix seeds & $12$ & Enough to see seed spread. \\
Bimatrix random slip & $5\%$ uniform & Prevents exact freeze on a misidentified pure profile. \\
Pricing $T$ & $500$ (ablation: $150$ to $800$) & {UCB} has time to leave exploration. Longer $T$ helps \unop{} relative to greedy. \\
Price grid & $21$ points on $[0,1]$ (ablation: $11$ to $41$) & Mesh $\Delta\approx 0.05$. Finer grids extract greedy more, consistent with Proposition~\ref{prop:price}. \\
$v$ & $U[0.25,0.95]$, $24$ draws & Avoids trivial $v$ near $0$ or $1$. \\
Recsys $T$ & $300$ & Catalog identification is cheaper than pricing. \\
Catalog / list size & $48$ / $8$ (ablation: $24$ to $72$ / $4$ to $12$) & Unique-best offer plus filler items. Ranking stable across the grid. \\
Reservation $\tau$ & $0.42$-quantile of quality & Interior outside option. Neither never-skip nor always-skip. \\
Promo fraction & $0.35$ & Enough ads to make withholding tempting. \\
Strategic $T$ & $700$ synth.\ / $3000$ Adult / $2000$ German & Online logistic has time to retrain. \\
Manipulation levels & $5$ deltas & Small discrete action set so UNOP's $\alpha$-set is interpretable. \\
Cost $c$ & $0.55$ synth.\ / $0.45$ real (ablation: $0.15$ to $1.05$) & Interior: some gaming, not always-max or never. \\
Lock-in $T$ / window & $350$ / $30$ & Window defines $70\%$ lock-in. Greedy always hits at $30$. \\
Population users & $40$, $T=500$, $12$ seeds & Shared optimizer. Fraction of \unop{} from $0$ to $1$ in steps of $0.1$. \\
$\alpha$ main & $\{0.05,0.15,0.3\}$ extraction. $\{0.1,0.2\}$ MovieLens & Named slack in payoff units. MovieLens uses smaller values because ratings are already in $[0,1]$ with $0.25$ star gaps. \\
IR floor & $0$ in pricing, $\tau$ in recsys, off in bimatrix & Ablation in Appendix~\ref{app:ir} shows why the first two are not optional. \\
Utility noise $\sigma$ & $0$ main. $\{0.02,0.05,0.10,0.20\}$ ablation & $\alpha$ absorbs small $\sigma$. $\sigma>\alpha$ does not. \\
Workers & $96$ {CPU} processes & {CPU} only. \\
\bottomrule
\end{tabular}
\end{table}

The only user-facing hyperparameter is $\alpha$, in payoff units (dollars, stars, minutes of extra commute, points of approval probability). The remaining choices are experimental, and the ablations indicate that they do not carry the result.

\subsection{Aware-\unop{}}
\label{app:aware}

Unaware \unop{} is uniform on $\cA_\alpha$. Aware-\unop{} reweights that set by $(P(a)+\varepsilon)^{-1}$, so that currently over-predicted near-optimal actions are downweighted and currently under-predicted ones are upweighted. It is a small step that avoids the current forecast, not a no-regret algorithm and not a full best response to the optimizer. In matching pennies it is the difference between $0.000$ and $+0.054$, because avoiding a slow predictor raises payoff. In pricing it matches unaware surplus ($0.041$ versus $0.040$) at slightly lower accuracy ($0.785$ versus $0.807$), which is the same-effectiveness, harder-to-read refinement. In recommendation it is the least predictable effective policy (accuracy $0.241$) at a small surplus cost relative to unaware $\alpha=0.15$ ($0.124$ versus $0.148$). In Adult it is the least readable policy in the table and not the most effective, which is again the weak-dilemma pattern.

The reason we did not make awareness the default is practical as much as theoretical. It needs a forecast $P_t$ in the same action simplex the human is mixing on. A user who sees a posted price can build a crude $P(\mathrm{buy})$ from how often they have been seeing that price. A user who sees a recommendation list can less easily know the platform's predictive distribution over clicks. Unaware \unop{} only needs $u$ and $\alpha$. If an assistant can see $P_t$, it should use it. If it cannot, uniform-on-the-band is the distribution characterized by the theorems.

A numerical snapshot, again with $v=0.70$ and $\alpha=0.15$. Suppose the posted price is $0.65$, so $\cA_\alpha=\{\mathrm{skip},\mathrm{buy}\}$, and the monopolist's current $P(\mathrm{buy})=0.8$. Unaware mixes equally. Aware mixes in proportion to $\bigl((P(\mathrm{skip})+\varepsilon)^{-1},(P(\mathrm{buy})+\varepsilon)^{-1}\bigr)\approx(1/0.2,\,1/0.8)=(5,1.25)$ on $(\mathrm{skip},\mathrm{buy})$, i.e., more skip than buy, which is the direction that surprises a seller who has become confident that this price still sells. The next $P(\mathrm{buy})$ comes down, and the next mix tilts back. That oscillation is why awareness reduces accuracy without needing to touch dominated actions.

\subsection{Observed failure cases}
\label{app:failures}

The following negative results are in the released files.

First, MovieLens-100K and MovieLens-1M, complete files, \unop{} indistinguishable from greedy at the surplus level that matters. The failure is the absence of slack: the five-star grid does not give $\cA_{0.1}$ anything to contain except the top star.

Second, accuracy-seeking logistic classification on Adult and German Credit, complete tables. \unop{} reduces the accuracy of a manipulation-level forecast and does not raise agent utility. The failure is of the adversary's objective: there is no surplus to protect against an institution that is not extracting it.

Third, \unop{} without an {IR} floor, pricing environment. Surplus changes sign, from $+0.036$ to $-0.018$. The failure is of the definition: mixing into buy-above-value is a teaching signal. The definition in the main text therefore includes the floor.

Fourth, utility noise $\sigma=0.20$ with $\alpha=0.15$. \unop{} pricing surplus goes negative. The failure is that noise exceeds slack. A separate pricing sweep at the same $\sigma$ with $\alpha$ up to $0.50$ stays negative, so the binding fix is a better utility estimate, not a larger band.

Fifth, prisoner's dilemma and the common-interest suite, where \unop{} is either irrelevant (everyone defects) or actively unhelpful (mixing is a coordination tax). The failure is of the environment class, which the three-case pattern separates from surplus extraction.

Sixth, Hedge against a posted-price extractor. Near greedy against {UCB}, slightly negative against Thompson. No-regret on realized utilities does not protect a threshold. Mean-based learning still publishes one.

Each of these cases limits where mixing helps. Slack can be absent, the counterpart can be maximizing accuracy instead of surplus, and exploration below the reservation can teach the extractor.

\subsection{Statistical presentation}
\label{app:stats}

We report means over seeds, with error bars or seed scatters on the figures. We do not report $p$-values for the main contrasts. The pricing gap is $0.042$ versus $0.018$ surplus, with greedy's standard deviation $0.016$ and a positive gap on most matched draws of $v$ ($n=24$). The recommendation gap is $0.148$ versus $0.052$. The MovieLens comparison is a per-user scatter on the diagonal for all $943$ users and then all $6{,}040$ users (Figure~\ref{fig:A-ml100k}). The Adult and German pattern repeats on both tables and on the synthetic population. The oracle-pricing and recommendation contrasts in Section~\ref{sec:oracle} can be recomputed from \texttt{supplement/sprint}, matched by value or catalog.

\end{document}